\documentclass[11pt]{article}
\usepackage[a4paper,margin=1.0in]{geometry}
\usepackage{cite}
\usepackage{amsmath,amssymb,amsfonts,amsthm,mathtools,bm}
\usepackage{graphicx}
\usepackage{subcaption}
\usepackage{makecell}
\usepackage{array}
\usepackage{booktabs}
\usepackage{siunitx}
\usepackage{algorithm}
\usepackage{algpseudocode}
\usepackage{hyperref}
\usepackage{cleveref}
\graphicspath{{./Figs/}}
\usepackage{float}
\usepackage{placeins}

\newtheorem{theorem}{Theorem}[section]
\newtheorem{remark}[theorem]{Remark}

\newcommand{\sat}{\operatorname{sat}}
\newcommand{\rlim}{\operatorname{rlim}}
\DeclareMathOperator{\rate}{rate}
\DeclareMathOperator{\sign}{sign}

\title{Sampled-data Robust Control of Electrically Stimulated Engineered Cell Factories}

\author{%
Papri Dey\textsuperscript{1},
Ksenia Zlobina\textsuperscript{1},
Nicholas A. Rondoni\textsuperscript{1},
Marcella M. Gomez\textsuperscript{1}\\[2mm]
\textsuperscript{1}Applied Mathematics, Baskin School of Engineering, University of California, Santa Cruz
\thanks{Corresponding author pdey@ucsc.edu}
}
\date{}
\begin{document}
\maketitle

\begin{abstract}
Closed-loop bioelectronic regulation of engineered secretory cell systems is challenging because electric-field (EF) stimulation acts indirectly through transcription-factor activation, in the presence of delayed, nonlinear, and noisy intracellular dynamics, sparse measurements, and constrained burst-based actuation. We develop a framework for robust closed-loop endocrine regulation in electrically stimulated engineered cell factories, illustrated through extracellular thyroid hormone \(T_4\) production in engineered thyroid-like cells. The plant is modeled by a control-oriented ODE formulation combining a reduced mechanistic \(T_4\) pathway, an EF-responsive Hill module, and a linear-chain Erlang cascade representing distributed intracellular delay. On this basis, we design a sampled-data adaptive proportional-integral-derivative (PID) controller with derivative filtering, anti-windup, saturation and rate limits, and hysteretic band-locking, together with a robust adaptive extension that accounts for parameter mismatch, sensor noise and bias, actuator mismatch, delay/jitter, and exogenous rhythmic disturbance through a scenario-based risk-aware update. We provide local sampled-data input-to-state stability interpretations for both APID and RAPID, showing that, under standard local Lyapunov and bounded-disturbance conditions, the sampled tracking error is ultimately bounded by a disturbance-dependent constant. In silico experiments demonstrate sustained regulation of extracellular \(T_4\) across prescribed targets despite significant uncertainty.
\end{abstract}

\section{Introduction}
Advances in synthetic biology and bioelectronic medicine make it increasingly natural to view engineered cellular systems as controllable production platforms rather than only as biological models \cite{Famm2013,Jafari2021FCArchitecture,Jafari2023Merging,Koutsouras2024,Selberg2020MLDriven}. In such systems, an external stimulus must be converted into a desired biological output despite nonlinear dynamics, delayed intracellular response, sparse sensing, and actuation constraints \cite{DelVecchio2018,Qian2018}. These features make closed-loop regulation essential for therapeutic and bioelectronic applications \cite{Olofsson2015,Preetam2024ElectricalWounds,Rolandi2019Electroceuticals}.

The control problem studied here is the regulation of extracellular thyroxine \(T_4^{\mathrm{ext}}\) in an electrically stimulated engineered thyroid-like cell factory. The motivating application is a bioelectronic implant for hypothyroidism, in which an onboard controller modulates EF stimulation to regulate hormone output. EF does not act directly on secretion, but through promoter activation and transcription-factor-mediated pathways, so the input-output dynamics are delayed, nonlinear, and uncertain \cite{Purvis2013}. In addition, stimulation is delivered through burst-based pulse hardware rather than as an arbitrary continuous signal, while extracellular measurements are sparse and may be corrupted by noise and bias.

To represent this setting, we develop a control-oriented \(16\)-state ODE model that combines a reduced thyroid hormone pathway with a data-informed EF-response module. The latter is modeled by a Hill-type activation law \cite{Alon2006} and an \(N\)-stage linear-chain cascade, yielding a distributed-delay approximation of intracellular response \cite{MacDonald1978TimeLags,HurtadoKirosingh2019GLCT}. We further incorporate a realistic burst-based actuator representation and adopt a sampled-data viewpoint: the plant evolves continuously, while sensing and control updates occur only at coarser time instants \cite{Astrom2010}. Open-loop simulations show that fixed EF amplitudes shift the operating range of extracellular \(T_4\) but do not robustly regulate it to a desired target.

To address this challenge, we design a sampled-data adaptive PID controller (APID) that updates once per stimulation window using sampled measurements. The controller includes derivative filtering, anti-windup correction, saturation and rate limits, and a hysteretic band-lock mechanism that maintains a learned basal stimulation level once the output approaches the target. We then develop a robust adaptive extension (RAPID) that accounts for simultaneous parameter mismatch, measurement noise and bias, actuator mismatch, delay/jitter in EF delivery, and exogenous rhythmic disturbance through a scenario-based mean-conditional value-at-risk (CVaR)  update. The resulting architecture is aimed at practical constrained regulation rather than unconstrained nominal tracking.

The choice of PID as the underlying controller structure is deliberate. Recent bioelectronic control studies have explored machine-learning and RBF-neural-network-based feedback architectures \cite{Zlobina2022RoleMLPrecisionMedicine,Marquez2025CellMigratoryPatterns}. These approaches offer higher expressive capacity, but they make it harder to isolate the effects of sampled sensing, burst-based EF delivery, and practical safeguards. In contrast, PID provides a transparent fixed-structure benchmark: the proportional, integral, and derivative actions are retained, and only the three gains \((K_p,K_i,K_d)\) are adapted online through a model-assisted one-window predictive update.

The contribution of this work does not lie in introducing PID control, derivative filtering, anti-windup mechanisms, or risk-aware optimization as isolated concepts, as these are already well established \cite{Astrom2006,AstromHagglund2001,Visioli2006}. Rather, its novelty lies in their integration within a single sampled-data feedback architecture specifically designed for delayed EF-driven endocrine regulation under burst-based actuation, sparse sensing, and multi-class uncertainty. In addition, we give local sampled-data input-to-state stability (ISS) interpretations for APID and RAPID, including the effect of derivative-on-measurement and, for RAPID, the descent structure of the scenario-based mean-CVaR gain update.

The remainder of the paper is organized as follows. \Cref{sec:background} develops the qualitative thyroid-cell model. \Cref{sec:actuator} introduces the burst-averaged EF actuator and open-loop analysis. \Cref{sec:adaptive_PID} presents APID, its local ISS interpretation, and nominal closed-loop results. \Cref{sec:rapid} develops RAPID, gives a local risk-aware sampled-data ISS interpretation, and compares its performance with APID under uncertainty. \Cref{sec:conclusion} concludes.

\section{Qualitative Model of Engineered Thyroid Cells}
\label{sec:background}
This section describes the engineered thyroid-cell system that serves as the biological plant in the proposed feedback architecture. The motivating application is a bioelectronic device that combines a sensor for extracellular \(T_4\), an EF actuator for stimulation, and a controller that uses the measured hormone level to regulate production toward a prescribed target. In this setting, the engineered cell culture is the hormone-producing subsystem, while the feedback loop couples sensing and actuation through the controller. We therefore begin with the biological pathway underlying \(T_4\) production and then derive a qualitative model suitable for control design. A system-level view of this sensing-actuation-control architecture is shown in Figure~\ref{fig:high_level_device}, which places the engineered thyroid-cell model in the broader context of the proposed implantable bioelectronic regulation platform.

\begin{figure}[t]
\centering
\includegraphics[width=0.9\linewidth]{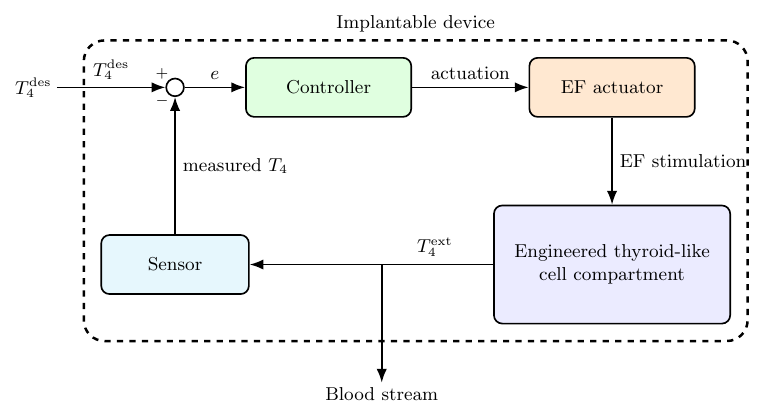}
\caption{System-level view of the proposed closed-loop thyroid-hormone regulation platform.}
\label{fig:high_level_device}
\end{figure}

\subsection{Application Setting and Biological Motivation}
\noindent \textbf{Biological basis of thyroid hormone production in engineered cultures:}
It has been shown that in vitro thyroid cell cultures can self-organize into organoid structures resembling thyroid follicles. This follicular architecture is essential for thyroxine ($T_4$) synthesis, as production depends on coordinated processes occurring across three distinct compartments: the follicle lumen, the surrounding cells, and the interstitial space between follicles \cite{Ogundipe2021,Liang2022}.

These processes are similar to the natural thyroid processes and are formalized in the schematic shown in Figure \ref{fig:FBNNs}. This schematic compresses the detailed thyroid-cell graphic into a three-compartment reaction map: region I represents the interstitial space in contact with blood (right, pink), region II the thyrocyte (middle, blue), and region III the follicular lumen/colloid (left, light pink). The numbers in Figure \ref{fig:FBNNs}) correspond to the following processes: thyroglobulin is secreted by the thyroid follicular cell (1) and released into the follicle colloid via exocytosis (2), where it binds iodine (3). Iodide enters the thyrocyte from blood through the Na/I transporter (4) and moves into the colloid via the pendrin transporter (5), where it is oxidized to iodine by thyroid peroxidase (6). The iodinated thyroglobulin is conjugated (7) and reabsorbed by endocytosis (8), then proteolytically processed to generate $T_4$ (9). $T_4$ can be converted to T3 by deiodinases (10), and both hormones are exported to blood via MCT transporters (11–14).

\begin{figure*}[ht!]
\centering
\includegraphics[width=0.75\linewidth]{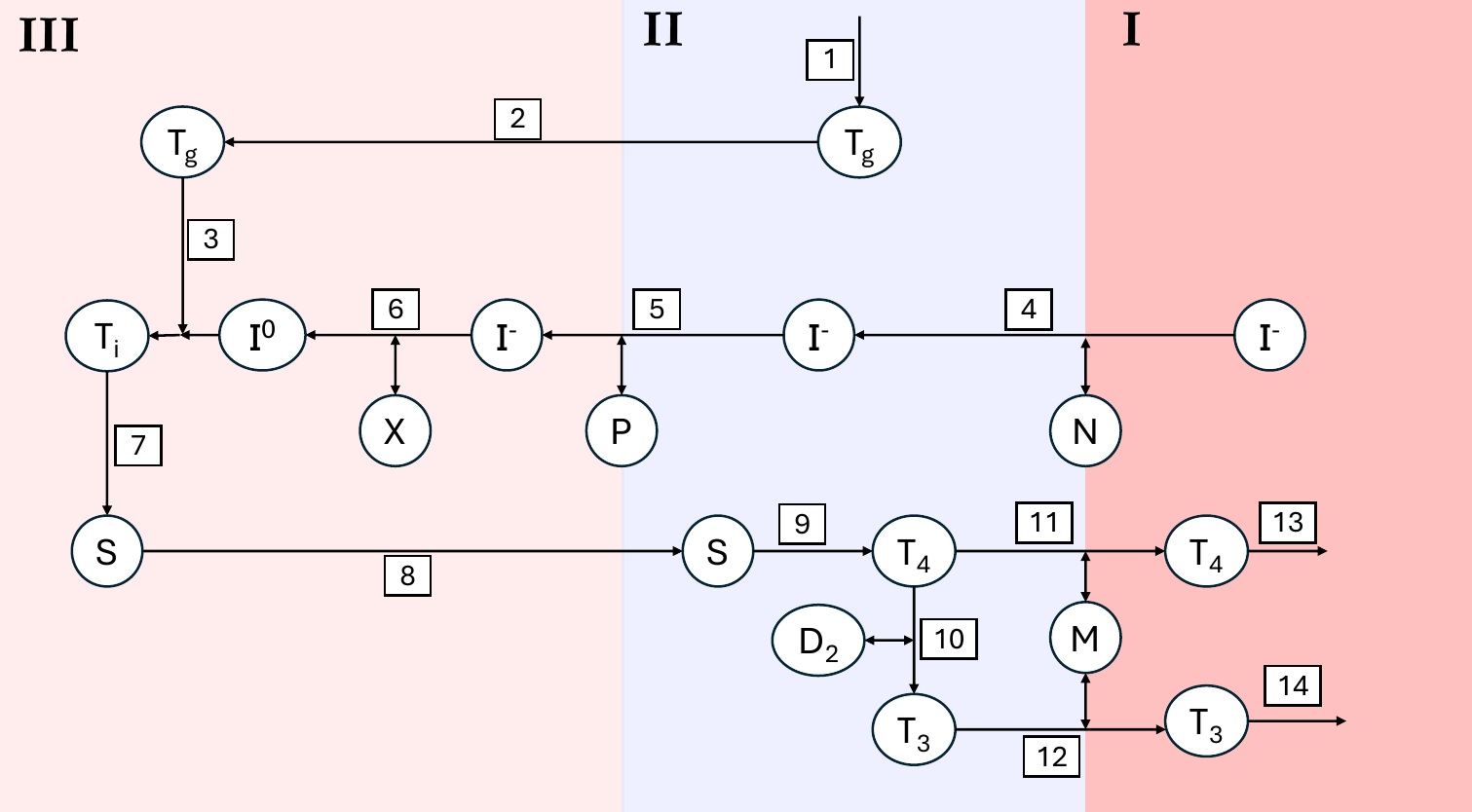}
\caption{Schematic of $T_4$ production pathways in the engineered cell culture. Tg – thyroglobulin, I – iodide, I0 – iodine, N – Na/I transporter. P – pendrin, X – thyroid peroxidase, Ti – iodinated thyroglobulin, S – conjugated complex, $T_4$ – thyroxine, T3 – triiodothyronine, D2 – deiodinase 2, M – MCT8 transporter.}
\label{fig:FBNNs}
\end{figure*}
It is important to identify which of these reactions represent the most suitable targets for therapeutic intervention. Some targets are relatively simple, for example, modulating iodine levels through dietary supplements, but such interventions have limited impact. More complex regulatory points can be accessed through modern genetic engineering. For instance, reaction 9, commonly referred to as ``proteolysis'' consists of multiple steps regulated by proteins expressed in thyroid follicular cells, and modulation of specific transcription factors can influence these processes. Thus, the ``proteolysis'' section of the reaction scheme represents a particularly promising target for therapeutic intervention or for intervention with a bioelectronic device.

\subsection{Reduced Scheme of \texorpdfstring{\(T_4\)}{} Production and Derivation of the Equations}
\label{sec:model}
We begin by constructing a control-oriented reduced model for \(T_4\) production based on its natural biochemical pathway. Using a chemical reaction network formulation, we derive a reduced mechanistic model that captures the core production, transport, and degradation processes relevant to extracellular hormone release. We then complement this mechanistic structure with data-informed parameter calibration so that the resulting model reflects the effective system time scales and delayed responses observed experimentally.
For control design, the full reaction scheme must be reduced to preserve the main actuation and output pathways while omitting biological details that do not materially affect controllability. Accordingly, reactions that do not directly influence the control input or measured output are merged into composite processes, yielding the simplified scheme shown in Figure~\ref{fig:model1}. This reduced representation retains the essential structure needed for model-based actuation and feedback design.
The simplified diagram (Figure \ref{fig:model1}) corresponds to the reactions in Figure \ref{fig:FBNNs} as follows. Reactions 1 and 2, representing the constant influx of thyroglobulin, are captured by reaction $\alpha_{T_g}$. Reactions 4–6, representing iodine influx, are combined as reaction $\alpha_{I}$. Reactions 3 and 7–9, describing thyroglobulin iodination and its stepwise conversion to $T_4$, are consolidated into reaction $k_1$, which constitutes the key regulatory node influenced at the transcriptional level. Consequently, $k_1$ represents the most effective point of intervention for a bioelectronic device.
Reaction 11, corresponding to $T_4$ release from the thyroid follicle into the bloodstream, is represented as MCT8 in the simplified diagram. Reactions 10, 12, and 14, which mediate $T_3$ production and release, are omitted under the assumption that $T_3$ and $T_4$ levels remain naturally balanced.
Accordingly, the mathematical model of the simplified system comprises the following reactions.

Thyroglobulin and iodine supply:
\begin{align*}
     \emptyset &\overset{\alpha_{T_g}}{\longrightarrow} T_g\\
    \emptyset &\overset{\alpha_I}{\longrightarrow} I
\end{align*}
$T_4$ production out of thyroglobulin and iodine:
\begin{align*}
&T_g +I \overset{k_1}{\longrightarrow} T_4^{int}
\end{align*}

Release of $T_4$ to blood flow - through MCT8 transporter:
\begin{align*}
&T_4^{int} \overset \omega {\longrightarrow} T_4^{ext}
\end{align*}

All substances participating in the process naturally degrade:
\begin{align*}
& T_g \overset{\gamma_{T_g}}{\longrightarrow} \emptyset\ \\
& I \overset{\gamma_I}{\longrightarrow} \emptyset\\
& T_4^{int} \overset{\gamma_{T_4}}{\longrightarrow} \emptyset \\
&T_4^{ext} \overset{\gamma_{T_4}}{\longrightarrow} \emptyset
\end{align*}
The overall reduced model of T\textsubscript{4} production in the engineered cell factory is represented by a system of four ordinary differential equations:
\begin{align} \label{eq:mechanistic_model}
    \begin{cases}
    \dot{I} &= -k_1 T_g I - \gamma_I I + \alpha_I, \\
    \dot{T_4^{\mathrm{int}}} &= k_1 T_g I - \omega T_4^{int} - \gamma_{T_4}^{\mathrm{int}}T_4^{int}, \\
    \dot{T_g} &= \alpha_{T_g} - k_1T_g I - \gamma_{T_g}T_g, \\
    \dot{T_4^{\mathrm{ext}}} &= \omega T_4^{int} -\gamma_{T_4}^{\mathrm{ext}}(t)T_4^{\mathrm{ext}} \\
    \end{cases}
\end{align}
    
\begin{figure}[htbp]
  \centering
  \begin{minipage}[t]{0.48\textwidth}
    \vspace{0pt}
    \includegraphics[width=\linewidth]{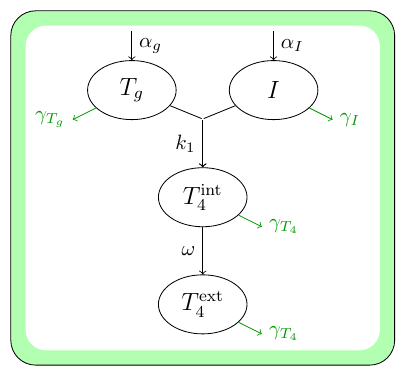}
  \end{minipage}
  \caption{ $T_4$ production pathway in engineered cell culture: schematic of the reduced model.}
  \label{fig:model1}
\end{figure}

\subsection{Effect of Electric Field on \texorpdfstring{$T_4$}{} production Model}
In the reduced model, $T_4$ production is represented as a single-step reaction between thyroglobulin and iodide. While this simplification captures the overall output, we know that in reality, the process involves multiple sequential stages, each regulated by various factors, including proteins whose roles are not yet fully elucidated. Advances in genetic engineering allow the integration of electric-field (EF)-responsive promoters into different genes, including those controlling these multi-stage reactions\cite{MansouriFussenegger2022}. One candidate is the transcription factor NKX2.1, whose importance in $T_4$ production has been established, although the precise mechanisms remain unclear \cite{Kusakabe2006,CarmonaHidalgo2024}. In our model, we assume that EF stimulation acts at the transcriptional level, initiating a cascade of molecular events that ultimately modulates the rate of $T_4$ synthesis from thyroglobulin and iodide (see Figure \ref{fig:cascade_scheme}). This provides a biologically plausible bridge between EF inputs and hormone output, which we formalize in the following ODE framework.

\begin{equation}
\begin{aligned}
\dot{u}_f    &= \alpha_1(t)\,\frac{EF(t)^n}{K^n + EF(t)^n} - \gamma_1\,u_f,\\
\dot{u}_1    &= a\,(u_f - u_1),\\
\dot{u}_i    &= a\,(u_{i-1} - u_i), \quad i=2,\dots,10,\\
\dot{u}_{NK} &= \alpha\, u_{10} - \beta\, u_{NK} + c.
\end{aligned}
\end{equation}
where $u_f$ represents the EF-induced promoter drive, modeled as a low-pass Hill function capturing the nonlinear activation of transcription by the applied electric field. The variables $u_1, \dots, u_{10}$ correspond to the stages of the N-stage linear (Erlang) cascade, which approximates the distributed delay in the cellular response due to transcription, translation, and intermediate regulatory processes \cite{Gomez2016StochasticDelay}. Finally, $u_{NK}$ denotes the cumulative effect of the entire cascade, the effective driver of $T_4$ production integrating the delayed dynamics of all upstream stages. Parameters $\alpha_1$, $K$, $n$, and $\gamma_1$ define the Hill nonlinearity, while $a$ sets the rate of the linear-chain cascade, $\alpha$ and $\beta$ are the synthesis and degradation rates of the cascade output, and $c$ represents any baseline activity independent of the EF input.

The output variable $u_{NK}$ modulates the kinetic parameter $k_1$ according to a Hill-type relation:
\[
k_1(u_{NK}) = \kappa \;\frac{\left(\frac{u_{NK}}{K_2}\right)^h}{1 + \left(\frac{u_{NK}}{K_2}\right)^h},
\]
where $\kappa$ is the maximal rate, $K_2$ is the half-activation constant, and $h$ is the Hill coefficient. This formulation captures how the cumulative effect of the EF-driven cascade influences the $T_4$ production pathway in a nonlinear, saturable manner.

\begin{figure*}[ht!]
\centering
\includegraphics[scale= 1.2]{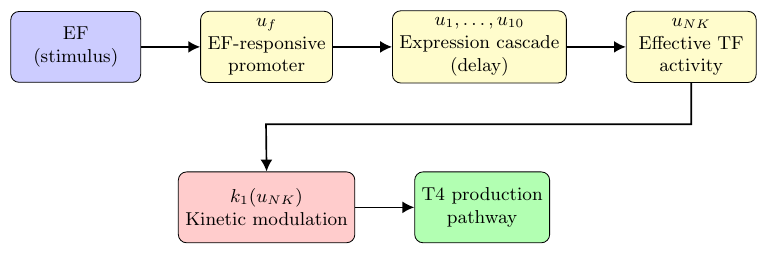}
\caption{EF-driven $T_4$ production cascade in engineered thyroid cells.  Electric field (EF) affects the synthetic EF-responsive promoter ($u_f$) and the multi-stage gene expression cascade ($u_1,\dots,u_{10}$) leading to effective transcription factor activity ($u_{NK}$). It modulates the kinetic parameter ($k_1$) and ultimately controls $T_4$ synthesis.}
\label{fig:cascade_scheme}
\end{figure*}
This cumulative cascade introduces a natural total delay in the system, often referred to as the translational or effective production delay, the time required to produce a functional protein from mRNA. This delay arises from multiple processes, including translation elongation, protein folding, and intermediate regulatory steps \cite{Gomez2016StochasticDelay}. To capture this in the mechanistic model, we implement a linear-chain (Erlang) delay appended to the upstream promoter input. The $N$-stage cascade approximates a gamma-distributed delay, with transfer function $H_d(s)=(a/(s+a))^{N}$ mean delay $N/a$, standard deviation $\sigma = \sqrt{N}/a$, and coefficient of variation $1/\sqrt{N}$. The corresponding impulse response is Erlang/Gamma distributed, yielding a biologically realistic representation of the distributed cellular response \cite{HurtadoKirosingh2019GLCT}.

\subsection{Qualitative Sixteen-State ODE Model: Integrating \texorpdfstring{$T_4$}{} Core Dynamics with EF-Induced Promoter Cascade}
\label{sec:16-state_model}
To describe the EF-driven $T_4$ production in the engineered thyroid cells, we combine the mechanistic (first-principle) model of the thyroid core with a data-driven representation of the delayed promoter cascade. This results in a sixteen-state ODE system, capturing both the upstream EF-induced transcriptional cascade and the downstream $T_4$ production dynamics.

We define the state vector as:
\[
y(t)=\big(I,T_4^{\mathrm{int}},T_g,T_4^{\mathrm{ext}},\ u_f,\ u_1,u_2,\dots,u_{10},u_{NK}\big)^\top,
\]
where $u_f$ represents the EF-responsive promoter, $u_1, \dots, u_{10}$ are the stages of the linear-chain cascade, $u_{NK}$ denotes the effective transcription factor activity, and the remaining states correspond to key biochemical species in the thyroid core.

This combined model integrates the mechanistic dynamics of $T_4$ synthesis with the distributed delays introduced by the upstream gene expression cascade. It provides a framework for simulating the system’s response to arbitrary EF inputs:

\begin{align} \label{governing_equation}
\begin{cases}
\dot{u}_f \; &=\; \alpha_1(t)\,\frac{EF(t)^n}{K^n + EF(t)^n} \;-\; \gamma_1\,u_f \\
\dot{u}_1 &= a\,(u_f - u_1), \\
\dot{u}_i &= a\,(u_{i-1} - u_i), \qquad i=2,\dots,10 \\
\dot{u}_{NK} &= \alpha\,u_{10} \;-\; \beta\,u_{NK} \;+\; c, \\
k_1(u_{NK})  &= \kappa\;\frac{\big(\tfrac{u_{NK}}{K_2}\big)^{h}}{1+\big(\tfrac{u_{NK}}{K_2}\big)^{h}} \\
\dot{I} &= -k_1(u_{NK}) T_g I - \gamma_I I + \alpha_I, \\
\dot{T_4^{\mathrm{int}}} &= k_1(u_{NK}) T_g I - \omega T_4^{int} - \gamma_{T_4}^{\mathrm{int}}T_4^{int}, \\
\dot{T_g} &= \alpha_{T_g} - k_1(u_{NK})T_g I - \gamma_{T_g}T_g, \\
\dot{T_4^{\mathrm{ext}}} &= \omega T_4^{int} -\gamma_{T_4}^{\mathrm{ext}}T_4^{\mathrm{ext}} 
\end{cases}
\end{align}
This system of equations is used as a surrogate for the biological system in subsequent controller testing. The parameter values employed in the simulations are listed in \Cref{tab:nominal_parameters}. The parameters were selected to reproduce the characteristic slow dynamics of thyroid hormone production; however, the model is not intended to provide quantitatively accurate physiological estimates, but rather to serve as a benchmark for controller evaluation.
\begin{table}[t]
\caption{Nominal plant parameters used in simulation.}
\label{tab:nominal_parameters}
\centering
\normalsize
\setlength{\tabcolsep}{4pt}
\begin{tabular}{lcccccccccc}
\toprule
EF cascade & $\alpha_1$ & $K$ & $n$ & $\gamma_1$ & $N$ & $a$ & $\alpha$ & $\beta$ & $c$ \\
Value & 3.308 & 1.195 & 2 & 0.0296 & 10 & 0.001 & 168.7624 & 0.01 & 0.25 \\
\midrule
Core model & $\kappa$ & $K_2$ & $h$ & $\alpha_I$ & $\gamma_I$ & $\alpha_{T_g}$ & $\gamma_{T_g}$ & $\gamma_{T_4}^{\rm int}$ & $\omega$ & $\gamma_{T_4}^{\rm ext}$ \\
Value & 1 & 1000 & 3 & 1 & 0.004 & 1 & 0.04 & 0.1504 & 0.15 & 0.01 \\
\bottomrule
\end{tabular}
\end{table}

\section{Actuator model: hardware-constrained EF delivery}
\label{sec:actuator}
In practice, EF cannot be applied as an arbitrary continuous signal. Most experimental and clinical stimulators deliver pulse-based waveforms with prescribed timing structure \cite{Steinhardt2024Pulsatile,Li2017Biphasic}. We therefore represent the actuator through a windowed burst-averaged input.

\subsection{Electric-field actuator: windowed burst-averaged representation}
The external electric field is represented as
\[
EF(x,t)=A_m\phi(x)s_m(t),\qquad t\in[t_m,t_m+T_s),
\]
where \(\phi(x)\) is a fixed spatial profile, \(A_m\) is the windowwise stimulation amplitude, \(s_m(t)\) describes the temporal switching pattern, and \(T_s\) is the control-window length. Within each window, stimulation is delivered through bursts of high-frequency micro-pulses. Because the biochemical dynamics evolve much more slowly than the micro-pulses, the plant is driven by a burst-averaged effective input rather than by the raw waveform.

Let \(T_p=t_1+t_2+t_3+t_4\) be the micro-pulse period and \(\delta_{\rm pulse}=(t_1+t_3)/T_p\) the pulse duty factor. If one burst contains \(n_p\) pulses, its active duration is \(t_{5b}=n_pT_p\). With inter-burst gap \(t_{\rm gap}\), define \(T_b=t_{5b}+t_{\rm gap}\) and choose
\[
n_b=\left\lfloor\frac{T_s-t_{6,\min}}{T_b}\right\rfloor_+,
\qquad T_{\rm on}=n_bT_b,
\qquad t_6=T_s-T_{\rm on}.
\]
Writing \(s=t-t_m\), the scalar effective input entering the ODE model is
\begin{equation}
\label{eq:ef_eff}
EF_{\rm eff}(t)=A_m\delta_{\rm pulse}
\sum_{j=0}^{n_b-1}{\bf 1}_{[jT_b,jT_b+t_{5b})}(s).
\end{equation}
where
\[
\mathbf{1}_{[a,b)}(s)=
\begin{cases}
1, & s\in[a,b),\\[4pt]
0, & s\notin[a,b).
\end{cases}
\]
Thus, the actuator is a sparse, piecewise-constant burst signal that equals \(A_m\delta_{\rm pulse}\) during active burst intervals and zero during off phases reserved for measurement and recovery.

The per-window amplitude is obtained from the controller command \(u_m\) by the thresholded affine map
\begin{equation}
\label{eq:clip}
\begin{aligned}
&A_m =
\begin{cases}
0,  \hspace{1.6in} u_m < \mathrm{thr},\\
\operatorname{sat}_{[A_{\min},A_{\max}]}
\!\left(k_A(u_m-\mathrm{thr})\right),
 u_m \ge \mathrm{thr},
\end{cases} \\
&\operatorname{sat}_{[a,b]}(x)
:= \min\{b,\max\{a,x\}\}.
\end{aligned}
\end{equation}
where $\mathrm{thr}$ is the activation threshold on the controller command $u_m$, and $A_{\min}$, $A_{\max}$ are lower and upper bounds on the stimulation amplitude $A_m$ respectively. The actuator is represented in the qualitative model through an effective, burst-averaged input, shown in \Cref{fig:am_vs_um_two_cases} and \Cref{fig:signal_expl}.
\begin{figure}[t]
\centering
\includegraphics[]{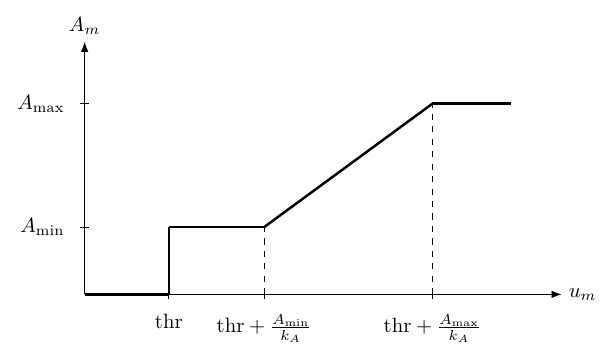}
\caption{Schematic of the amplitude map $A_m$ versus the controller command $u_m$. The two parameter orderings $A_{\min} < \mathrm{thr} < A_{\max}$ and $\mathrm{thr}< A_{\min}< A_{\max}$ are qualitatively similar, though not exactly identical: in both cases, the map is zero below $\mathrm{thr}$, then rises with slope $k_A$, and finally saturates at $A_{\max}$.}
\label{fig:am_vs_um_two_cases}
\end{figure}

\begin{figure}
    \centering
    \includegraphics[width=0.9\linewidth]{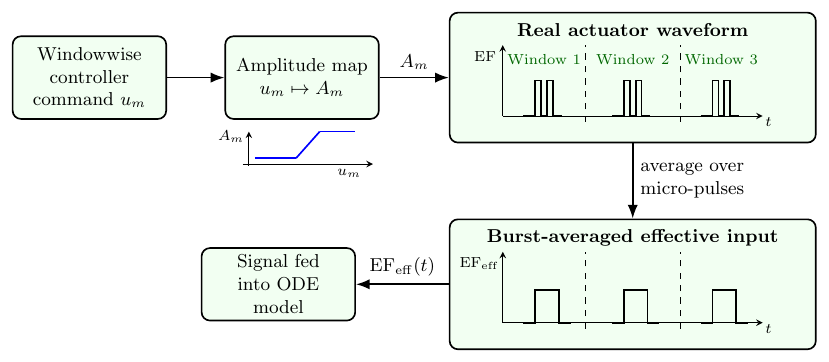}
    \caption{Hardware-to-model actuator mapping. A windowwise command \(u_m\) is converted to an amplitude \(A_m\), delivered physically through burst trains of micro-pulses, and represented in the ODE model by the burst-averaged effective input \(\mathrm{EF}_{\mathrm{eff}}(t)\).}
    \label{fig:signal_expl}
\end{figure}
Additionally, \Cref{fig:ef-supplied-zoom} illustrates the actuator actually seen by the qualitative model on a representative subset of the time horizon. The signal is sparse in time, with clear alternation between active burst epochs and zero-input intervals, and its magnitude increases stepwise from one window to the next as the controller raises the applied amplitude. This representation preserves the experimentally relevant timing structure of stimulation while remaining tractable for simulation and feedback design.

\begin{figure}[htbp]
    \centering
        \includegraphics[width=\linewidth]{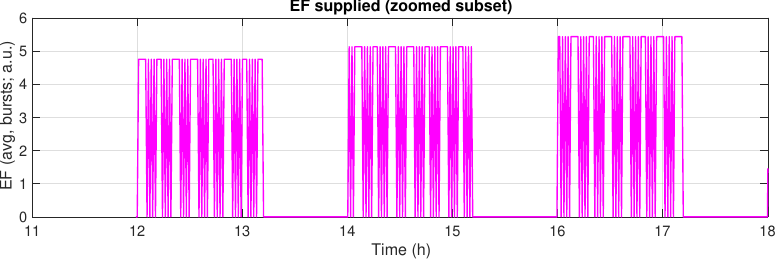}
        \caption{Burst-averaged EF supplied to the plant over a zoomed time interval. Narrow pulse blocks denote active stimulation windows; zero intervals denote off phases for measurement and recovery.}
        \label{fig:ef-supplied-zoom}
\end{figure}
The burst-based EF structure is treated primarily as a hardware constraint rather than as the source of biological delay: the stimulator delivers input through prescribed windows, bursts, and micro-pulses, whereas the dominant delay in hormone regulation arises from intracellular transcriptional and translational processes, modeled separately through the EF-responsive promoter module and the linear-chain/Erlang cascade. 
A numerical time-resolution study of this actuator averaging is presented in \Cref{subsec:time-resolution} after the closed-loop PID methodology.
\subsection{Open-loop analysis and sensitivity to EF amplitude}

We first examine the open-loop input-output behavior of the 16-state plant under fixed per-window EF amplitudes. In this study, the burst-averaged actuator signal is propagated through the model with constant windowwise amplitude
\[
A \in \{0,\;0.005,\;0.01,\;0.015,\;0.02\},
\]
and the resulting trajectories of \(T_4^{\mathrm{ext}}(t)\) are simulated over an \(80\)-hour horizon; see \Cref{fig:openloop_fixed_amp}. The responses show a clear amplitude dependence: increasing \(A\) shifts the output to a higher oscillatory regime, with \(A=0\) remaining near the basal level and larger amplitudes producing substantially elevated hormone levels. However, the effect is not merely a smooth vertical shift. The transient rise, oscillation magnitude, and long-time operating level all change noticeably with relatively small changes in \(A\), indicating that the output is quite sensitive to the EF amplitude, shown in \Cref{fig:openloop_fixed_amp}.

\begin{figure}[t]
    \centering
    \includegraphics[width=0.8\linewidth]{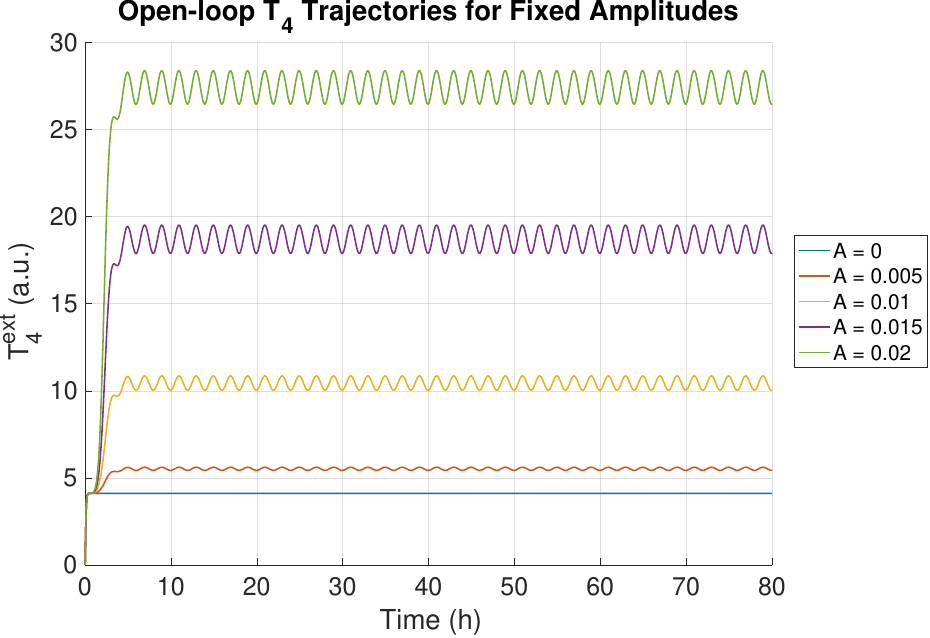}
    \caption{Open-loop trajectories of \(T_4^{\mathrm{ext}}(t)\) under fixed EF amplitudes \(A \in \{0,0.005,0.01,0.015,0.02\}\). Increasing amplitude shifts the response to higher oscillatory regimes, while also changing the transient rise and oscillation magnitude, highlighting the sensitivity of the output to EF amplitude.}
    \label{fig:openloop_fixed_amp}
\end{figure}

\noindent This sensitivity is important from a control standpoint. Small amplitudes produce only modest elevation above baseline, whereas slightly larger amplitudes can lead to much faster growth and more pronounced oscillations. Thus, even before introducing uncertainty or measurement limitations, the plant already exhibits a nontrivial amplitude-response dependence: the same actuator parameter that raises the mean hormone level also changes the transient and oscillatory characteristics of the response. In this sense, the open-loop study identifies not only the rough actuation range of the system, but also the sensitivity that any feedback law must handle. The controller developed in the next section is therefore tasked with regulating a delayed nonlinear system whose output can change appreciably under comparatively small variations in the applied EF amplitude.
\section{Adaptive PID (APID) Control Design}
\label{sec:adaptive_PID}
\Cref{sec:adaptive_PID} develops the sampled-data adaptive PID controller used for closed-loop regulation of \(T_4^{\mathrm{ext}}\). The controller acts once per window, maps the windowwise EF command to the actuator amplitude defined in \Cref{sec:actuator}, updates the PID gains through a one-window predictive cost, and switches to a hysteretic basal holding law after target entry. In this way, the section connects the actuator representation of \Cref{sec:actuator} with a practically implementable feedback rule for delayed endocrine regulation. 

\noindent The APID implementation is summarized in \Cref{alg:apid_compact} and schematically in \Cref{fig:apid}.
The algorithm combines sampled-data PID feedback, finite-difference retuning of
\(\phi=[K_p,K_i,K_d]^\top\), actuator saturation/rate limits, and a hysteretic
band-lock holding mode after target entry.

\begin{figure}[htbp]
\centering
\includegraphics[width=\linewidth]{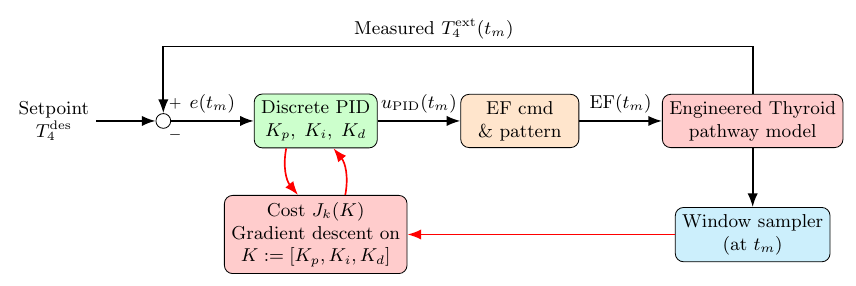}
\caption{High-level schematic of the APID controller. The feedback loop generates the EF command, while the adaptation loop updates the PID gains from a one-window predictive cost.}
\label{fig:apid}
\end{figure}
We now introduce the sampled-data adaptive PID controller used for closed-loop regulation of extracellular thyroid hormone. The control objective is to drive
\[
z(t)=T_4^{\mathrm{ext}}(t)
\]
toward a prescribed constant setpoint \(T_4^{\mathrm{des}}\) by adjusting the windowwise EF actuation described in \Cref{sec:actuator}. 
\subsection{Control objective and sampled-data architecture}
\label{subsec:architecture}

The control objective is to regulate the extracellular hormone output
\[
z(t)=T_4^{\mathrm{ext}}(t)=y_4(t)
\]
toward the constant target \(T_4^{\mathrm{des}}\). The plant evolves continuously, while the controller updates only at the window boundaries
\[
t_m=t_0+mT_s,\qquad m=0,1,2,\dots.
\]
At the \(m\)-th update, the sampled output is
\[
z_m:=z(t_m)=T_4^{\mathrm{ext}}(t_m),
\]
and the sampled tracking error is defined by
\[
e_m:=T_4^{\mathrm{des}}-z_m.
\]

\noindent The controller maintains an integral state
\[
I_m=I_{m-1}+e_mT_s
\]
and a filtered derivative-on-measurement term
\[
d_m=\frac{z_m-z_{m-1}}{T_s},\qquad
d_m^{(f)}=(1-\alpha_d)d_{m-1}^{(f)}+\alpha_d d_m.
\]
\noindent The windowwise proportional, integral, and derivative gains are collected into
\[
\phi_m=
\begin{bmatrix}
K_p^{(m)} & K_i^{(m)} & K_d^{(m)}
\end{bmatrix}^{\!\top}.
\]

\noindent Given \(\phi_m\), the unsaturated PID command is formed as
\[
u_{\mathrm{PID},m}:=EF_m^{\mathrm{unsat}}
=
K_p^{(m)}e_m
+
K_i^{(m)}I_m
-
K_d^{(m)}d_m^{(f)}.
\]
Let $EF_m$ denote the EF command currently applied over the window  $[t_m,t_m+T_s)$, with $EF_0$ prescribed at initialization. The PID law computes a new unconstrained command $EF_m^{\mathrm{unsat}}$, and the next applied command $EF_{m+1}$ is obtained by rate limiting and saturation relative to the current command $EF_m$ as follows. 
Define
\[
\rlim_{\Delta}(x,x_{\rm prev})
:=x_{\rm prev}+\sat_{[-\Delta,\Delta]}(x-x_{\rm prev}).
\]
The next EF command is
\begin{equation}
\label{eq:rate-limit}
EF_{m+1}=\sat_{[EF_{\min},EF_{\max}]}
\left(\rlim_{\Delta EF_{\max}}(EF_m^{\rm unsat},EF_m)\right).
\end{equation}
where $EF_{\min}$ and $EF_{\max}$ denote the admissible lower and upper bounds on the windowwise EF command, while $\Delta EF_{\max}$ is the per-window rate-limit bound.
The command is subsequently mapped to \(A_{m+1}\) using \eqref{eq:clip}, together with amplitude rate limiting and bounds,
so that the command can change by at most $\Delta EF_{\max}$ units between consecutive control updates.
Thus, the controller acts only through the scalar windowwise quantities \(EF_m\) and \(A_m\), while the plant receives the corresponding burst-averaged EF signal within each control window.

\noindent The PID gains $\phi=\begin{bmatrix}
K_p & K_i & K_d \end{bmatrix}^{\!\top}$ are not kept fixed. 
For a candidate gain vector \(\phi=[K_p,K_i,K_d]^\top\), the one-window predictive
cost is evaluated by applying the corresponding sampled PID update from the
current sampled state and propagating the model over \([t_m,t_m+T_s]\). If
\(z_\phi(t)\) denotes the resulting predicted output and \(A_\phi\) the trial
amplitude, we use
\begin{equation*}
\label{eq:one_window_cost}
J_m(\phi)
=
\int_{t_m}^{t_m+T_s} 
\beta_{\mathrm{ov}}[e_\phi(t)]_+^2+[e_\phi(t)]_-^2 dt+\rho_A A_\phi^2, \\
\end{equation*}
where \(e_\phi(t)=z_\phi(t)-T_4^{\mathrm{des}}\),
\([x]_+=\max\{x,0\}\), \([x]_-=\min\{x,0\}\),
\(\beta_{\mathrm{ov}}\ge1\) penalizes overshoot, and
\(\rho_A\ge0\) penalizes actuation effort.
The gradient \(\nabla_\phi J_m\) is approximated by bounded central finite differences, and the resulting gain update is projected to admissible bounds and smoothed before being used in the next window. In this way, the controller preserves the fixed PID structure while retuning its three gains online to match the delayed nonlinear endocrine dynamics.

\begin{remark}
The derivative term is implemented as derivative-on-measurement rather than derivative-on-error. Since the target \(T_4^{\mathrm{des}}\) is constant, one has \(\dot e(t)=-\dot z(t)\), so the two formulations are equivalent up to sign. Using derivative-on-measurement avoids the standard derivative-kick effect associated with differentiating the tracking error directly and is consistent with the filtered sampled-data implementation adopted here.
\end{remark}
\begin{remark}
In the implemented controller, this sampled-data adaptive PID law is further combined with practical safeguards including effective integral-gain scheduling near the target, anti-windup
back-calculation, and EF/amplitude rate limits. 
\end{remark}
\begin{algorithm}[!htbp]
\footnotesize
\caption{Sampled-Data Adaptive PID with Band-Lock Holding}
\label{alg:apid_compact}
\begin{algorithmic}[1]
\Require Setpoint \(T_4^{\mathrm{des}}\), window length \(T_s\), initial gains
\(\phi_0=[K_p,K_i,K_d]^\top\), initial command \(EF_0\), amplitude \(A_0\), actuator bounds, gain bounds, and band-lock parameters.
\Ensure Sampled output, EF command, amplitude, gain, cost, and lock histories.

\State Initialize \(I_0=0\), \(d_0^{(f)}=0\), \(z_{-1}=z(t_0)\), \(A_{\mathrm{basal}}=A_0\), and \(\texttt{locked}=0\).
\For{each control window \(m=0,1,2,\ldots\)}
    \State Apply the burst-averaged actuator input \(EF_{\mathrm{eff}}(t)\) generated by the held amplitude \(A_m\) on \([t_m,t_{m+1})\).
    \State Propagate the plant continuously and sample \(z_m=T_4^{\mathrm{ext}}(t_m)\).
    \State Compute
    \[
    e_m=T_4^{\mathrm{des}}-z_m,
    d_m=\frac{z_m-z_{m-1}}{T_s},
    d_m^{(f)}=(1-\alpha_d)d_{m-1}^{(f)}+\alpha_d d_m .
    \]
    \State Update the hysteresis state: enter lock if
    \(|e_m|\le \eta_{\mathrm{in}}T_4^{\mathrm{des}}\), and release lock if
    \(|e_m|\ge \eta_{\mathrm{out}}T_4^{\mathrm{des}}\), with
    \(0<\eta_{\mathrm{in}}<\eta_{\mathrm{out}}\).
    \If{\(\texttt{locked}=1\)}
        \State Update the basal amplitude and apply the holding law
        \[
        A_{\mathrm{basal}}\leftarrow
        \operatorname{sat}_{[A_{\min},A_{\max}]}
        (A_{\mathrm{basal}}+\eta_{\mathrm{basal}}e_m),
        \]
        \[
        A_{m+1}\leftarrow
        \operatorname{sat}_{[A_{\min},A_{\max}]}
        (A_{\mathrm{basal}}+k_{\mathrm{lock}}e_m).
        \]
        \State Recover \(EF_{m+1}\) from the inverse amplitude map and apply the EF and amplitude rate limits.
    \Else
        \State Evaluate the one-window predictive cost \(J_m(\phi_m)\) and approximate
        \[
        \nabla_\phi J_m(\phi_m)
        \approx
        \left[
        \frac{J_m(\phi_m+\varepsilon_j e_j)-J_m(\phi_m-\varepsilon_j e_j)}
        {2\varepsilon_j}
        \right]_{j=1}^3 .
        \]
        \State Update and project the PID gains:
        \[
        \phi_{m+1}
        =
        \Pi_{\Phi}\!\left(
        \phi_m-\alpha_\phi
        \frac{\nabla_\phi J_m(\phi_m)}
        {\max\{1,\|\nabla_\phi J_m(\phi_m)\|\}}
        \right).
        \]
        where \(\alpha_{\phi}=[\alpha_p,\alpha_i,\alpha_d^{(g)}]^\top\) is the componentwise gain-adaptation stepsize vector for \(\phi=[K_p,K_i,K_d]^\top\).
        \State Update the integral state \(I_m\leftarrow I_{m-1}+e_mT_s\), with anti-windup correction after saturation.
        \State Form the unsaturated PID command
        \[
        EF_m^{\mathrm{unsat}}
        =
        K_p^{(m+1)}e_m
        +
        K_i^{(m+1)}I_m
        -
        K_d^{(m+1)}d_m^{(f)} .
        \]
        \State Apply EF rate limiting and saturation:
        \[
        EF_{m+1}
        =
        \operatorname{sat}_{[EF_{\min},EF_{\max}]}
        \!\left(
        \operatorname{rlim}_{\Delta EF_{\max}}
        (EF_m^{\mathrm{unsat}},EF_m)
        \right).
        \]
        \State Map \(EF_{m+1}\) to \(A_{m+1}\) using the thresholded affine amplitude map and apply amplitude bounds/rate limits.
    \EndIf
    \State Log \(z_m,EF_{m+1},A_{m+1},\phi_{m+1},J_m\), and the lock state.
\EndFor
\end{algorithmic}
\end{algorithm}
In \Cref{alg:apid_compact}, \(J_m(\phi)\) denotes the one-window
predictive tracking-and-actuation cost evaluated from the current sampled state,
\(\Pi_\Phi\) denotes projection onto admissible gain bounds, and
\(\operatorname{rlim}_{\Delta}\) is the rate limiter defined in
\eqref{eq:rate-limit}.
\subsection{One-Window Predictive Cost Evaluation}
Algorithm~\ref{alg:window_cost} defines the one-window predictive cost used to retune the PID gains. For a candidate gain vector $\phi=(K_p,K_i,K_d)$, the algorithm constructs the corresponding trial EF command, applies the same actuator constraints used in the closed-loop simulation, maps the command to a burst-averaged amplitude, and propagates the plant over one control window. The returned quantity $J_m(\phi)$ measures the predicted tracking error over the next window together with a quadratic actuation penalty. This cost is used only for model-assisted gain adaptation; it is not a separate controller acting on the plant.
\begin{algorithm}[H]
\small
\caption{WindowCost: One-Window Predictive Cost Evaluation}
\label{alg:window_cost}
\begin{algorithmic}[1]
\Require
Current state $y(t_m)$, candidate gains $\phi=(K_p,K_i,K_d)$, setpoint
$T_4^{\mathrm{des}}$, samples $z_m,z_{m-1}$, derivative state
$d_{m-1}^{(f)}$, integral state $I_m$, previous EF command $EF_m$,
current amplitude $A_m$, basal amplitude $A_{\mathrm{basal}}$,
band-entry flag, controller parameters, amplitude-map parameters, and
plant/actuator parameters.
\Ensure
Predictive one-window cost $J_m(\phi)$.

\State Compute
\[
e_m=T_4^{\mathrm{des}}-z_m,
d_m=\frac{z_m-z_{m-1}}{T_s},
d_m^{(f)}=(1-\alpha_d)d_{m-1}^{(f)}+\alpha_d d_m .
\]

\State Set
\[
K_i^{\mathrm{eff}}=
\begin{cases}
1.25K_i, & |e_m|<e_{\mathrm{small}},\\
K_i, & |e_m|\ge e_{\mathrm{small}},
\end{cases}
\qquad
I_m^{\mathrm{tr}}=I_m+e_mT_s .
\]

\State Form the trial unsaturated command
\[
EF_m^{\mathrm{unsat}}
=
K_p e_m+K_i^{\mathrm{eff}}I_m^{\mathrm{tr}}
-K_d d_m^{(f)} .
\]

\State Apply EF rate limiting and saturation:
\[
EF_m^{\mathrm{tr}}
=
\sat_{[EF_{\min},EF_{\max}]}
\!\left(
\rate_{\Delta EF_{\max}}(EF_m^{\mathrm{unsat}},EF_m)
\right).
\]

\State Apply anti-windup:
\[
I_m^{\mathrm{tr}}
\leftarrow
I_m^{\mathrm{tr}}
+
K_{b,\mathrm{aw}}
\bigl(EF_m^{\mathrm{tr}}-EF_m^{\mathrm{unsat}}\bigr)T_s .
\]

\State Map $EF_m^{\mathrm{tr}}$ to amplitude and apply smoothing, rate limits, and bounds:
\State Set $\widetilde A_m^{\mathrm{tr}}
=(1-\lambda_A)A(EF_m^{\mathrm{tr}})+\lambda_A A_m$.

\State Update
\[
A_m^{\mathrm{tr}}
\leftarrow
\sat_{[A_{\min},A_{\max}]}
\!\left(
\rate_{\Delta A_{\max}}
(\widetilde A_m^{\mathrm{tr}},A_m)
\right).
\]

\If{the band-entry flag is active}
\State Enforce the post-band soft floor:
\[
A_m^{\mathrm{tr}}
\leftarrow
\max\{A_m^{\mathrm{tr}},\,0.98A_{\mathrm{basal}}\}.
\]
\EndIf

\State Simulate the plant on $[t_m,t_m+T_s]$ using the burst-averaged EF profile generated by $A_m^{\mathrm{tr}}$, and let $z_\phi(t)$ denote the predicted output.

\State Define
\[
\begin{aligned}
J_m(\phi)
&=
\int_{t_m}^{t_m+T_s}
\left(
\beta_{\mathrm{ov}}[e_\phi(t)]_+^2
+
[e_\phi(t)]_-^2
\right)\,dt
+
\rho_A\bigl(A_m^{\mathrm{tr}}\bigr)^2,\\
e_\phi(t)&=z_\phi(t)-T_4^{\mathrm{des}},
\quad
[x]_+=\max\{x,0\},\quad [x]_-=\min\{x,0\}.
\end{aligned}
\]

\State \Return $J_m(\phi)$.
\end{algorithmic}
\end{algorithm}
\noindent \Cref{alg:window_cost} evaluates the one-window predictive cost by replaying the same controller-side APID update ingredients used in the live loop, including filtered derivative-on-measurement, effective integral-gain scheduling, EF rate limiting, anti-windup correction, amplitude rate limiting, and the post-band amplitude floor after first target entry. 
The weights \(\beta_{ov}\) and \(\rho_A\) are cost-shaping parameters in the one-window predictive objective. In practice, increasing \(\beta_{ov}\) makes the controller more cautious near the target by discouraging overshoot, whereas increasing \(\rho_A\) reduces actuation aggressiveness at the cost of slower convergence.
\subsection{Sampled-Data adaptive PID with band-lock stabilization}
\label{subsec:band_lock}
Because the plant is delayed and the actuation is windowed, a purely adaptive PID update can still exhibit late-time drift or repeated correction after the trajectory first reaches the target neighborhood. To reduce this effect, we augment the adaptive PID law with a hysteretic band-lock mechanism.

\noindent Let
\( 
e_m^{\mathrm{true}}:=T_4^{\mathrm{des}}-z_m,
\)
and choose inner and outer hysteresis radii
\[
\delta_{\mathrm{in}}=\eta_{\mathrm{in}}T_4^{\mathrm{des}},
\qquad
\delta_{\mathrm{out}}=\eta_{\mathrm{out}}T_4^{\mathrm{des}},
\qquad
0<\eta_{\mathrm{in}}<\eta_{\mathrm{out}}.
\]
When the sampled output first satisfies $|e_m^{\mathrm{true}}|\le \delta_{\mathrm{in}}$,
the controller declares band entry, stores the current amplitude as a basal level \(A_{\mathrm{basal}}\), resets the integral state, and switches from the normal adaptive PID update to a low-gain holding law. While the trajectory remains locked, the basal amplitude is updated slowly according to the sign of the tracking error, and the next commanded amplitude is taken as a small correction around this basal level:
\[
A_{m+1}
=
\operatorname{sat}\!\left(
A_{\mathrm{basal}}+k_{\mathrm{lock}}e_m^{\mathrm{true}},
A_{\min},A_{\max}
\right),
\]
followed by the same amplitude rate limit as in the controller. The corresponding EF command is then recovered from the inverse of the amplitude map.

\noindent If the output later leaves the larger outer band,  $|e_m^{\mathrm{true}}|\ge \delta_{\mathrm{out}},$
the lock is released and the controller returns to the full adaptive PID update. The use of distinct entry and exit radii creates hysteresis and prevents rapid switching near the boundary.

\noindent This procedure is reported in \Cref{alg:bandlock_update} and the effect of this mechanism is visible in \Cref{fig:apid_bandlock_setpoint25}, where the response approaches the \(\pm 5\%\) band and then remains near the target without the large late-time drop seen in the no-band-lock comparison reported in \Cref{fig:bandlock_vs_nobandlock}.

Algorithm~\ref{alg:bandlock_update} describes the low-gain holding law used after the measured output first enters the prescribed target band. Rather than continuing to adapt the PID gains aggressively near the setpoint, the controller stores a basal amplitude and applies small corrections around this learned sustaining level. The hysteretic band-lock mechanism reduces late-time drift and repeated over-correction in the delayed endocrine response, while retaining the same amplitude bounds, rate limits, and EF-command constraints used by the normal controller.
\begin{algorithm}[H]
\footnotesize
\caption{\textsc{BandLockUpdate}: Basal Holding Law Near the Setpoint}
\label{alg:bandlock_update}
\begin{algorithmic}[1]
\Require
True setpoint error $e_m^{\mathrm{true}}=T_4^{\mathrm{des}}-z_m$,
current amplitude $A_m$, current EF command $EF_m$, and basal amplitude
$A_{\mathrm{basal}}$.
\Ensure
Next amplitude $A_{m+1}$, next EF command $EF_{m+1}$, and updated basal
amplitude $A_{\mathrm{basal}}$.

\State Update the basal support amplitude:
\[
A_{\mathrm{basal}}
\leftarrow
\sat_{[A_{\min},A_{\max}]}
\!\left(
A_{\mathrm{basal}}+\eta_{\mathrm{basal}}e_m^{\mathrm{true}}
\right).
\]

\State Form a low-gain correction around the basal level:
\[
A_{m+1}
\leftarrow
\sat_{[A_{\min},A_{\max}]}
\!\left(
A_{\mathrm{basal}}+k_{\mathrm{lock}}e_m^{\mathrm{true}}
\right).
\]

\State Apply amplitude rate limiting and bounds:
\[
A_{m+1}
\leftarrow
\sat_{[A_{\min},A_{\max}]}
\!\left(
\rate_{\Delta A_{\max}}(A_{m+1},A_m)
\right).
\]

\State Recover the inverse-map EF command:
\[
EF_{m+1}^{\mathrm{inv}}
=
\begin{cases}
0, & A_{m+1}\le 0,\\[1mm]
\mathrm{thr}+A_{m+1}/k_A, & A_{m+1}>0.
\end{cases}
\]

\State Apply EF rate limiting and saturation:
\[
EF_{m+1}
\leftarrow
\sat_{[EF_{\min},EF_{\max}]}
\!\left(
\rate_{\Delta EF_{\max}}(EF_{m+1}^{\mathrm{inv}},EF_m)
\right).
\]

\State \Return $(A_{m+1},EF_{m+1},A_{\mathrm{basal}})$.
\end{algorithmic}
\end{algorithm}
\subsection{Normal Adaptive PID Gain Update}
Algorithm~\ref{alg:normal_adaptive_PID} gives the normal APID update used outside the locked holding mode. The algorithm approximates the gradient of the one-window cost $J_m(\phi)$ by central finite differences, takes a bounded componentwise descent step, projects the gains onto admissible intervals, and then computes the sampled-data PID command. The remaining steps enforce derivative-on-measurement feedback, integral leakage and anti-windup, EF saturation and rate limits, amplitude bounds, and the post-band amplitude floor. Thus, the update preserves a standard PID feedback structure while adapting only the three gains $\phi_m=[K_p,K_i,K_d]^T$. The vector
\[
\alpha_{\phi}
=
[\alpha_p,\alpha_i,\alpha_d^{(g)}]^T
\]
collects the componentwise adaptation stepsizes for
\(\phi_m=[K_p,K_i,K_d]^T\), where \(\alpha_d^{(g)}\) denotes the adaptation stepsize for \(K_d\) and is distinct from the derivative-filter coefficient \(\alpha_d\).
\begin{algorithm}[H]
\footnotesize
\caption{\textsc{AdaptivePIDUpdate}: Normal Adaptive PID Law}
\label{alg:normal_adaptive_PID}
\begin{algorithmic}[1]
\Require
State $y(t_m)$, reference $r_f(t_m)$, samples $z_m$, errors
$e_m=r_f(t_m)-z_m$, $e_{m-1}$, derivative state $d_m^{(f)}$,
integral state $I_m$, command $EF_m$, amplitude $A_m$, basal amplitude
$A_{\mathrm{basal}}$, and flag $\texttt{first\_band\_hit}$.
\Ensure
$\phi_{m+1}$, $I_m$, $EF_{m+1}$, and $A_{m+1}$.

\State Set
\begin{eqnarray*}
\phi_m&=[K_p,K_i,K_d]^T,\quad
\alpha_\phi=[\alpha_p,\alpha_i,\alpha_d^{(g)}]^T,\quad \\
&\Delta\phi_{\max}=[\Delta K_p^{\max},\Delta K_i^{\max},
\Delta K_d^{\max}]^T .
\end{eqnarray*}

\State For $j=1,2,3$, compute
\[
[g_m]_j =[\nabla_{\phi}J_m(\phi_m)]_{j}
\approx
\frac{
J_m(\phi_m+\varepsilon_{j,m}e_j)
-
J_m(\phi_m-\varepsilon_{j,m}e_j)}
{2\varepsilon_{j,m}},
\]
\State Normalize the gradient safely: $\|g_m\|_{\mathrm{safe}}:=\max\{1,\|g_m\|_2\}$.
\State Take a bounded componentwise descent step:
\[
\Delta\phi_m
=
\sat_{[-\Delta\phi_{\max},\,\Delta\phi_{\max}]}
\!\left(
-\alpha_\phi\odot \frac{g_m}{\|g_m\|_{\rm safe}}
\right).
\]

\State Project and smooth the gain update:
\[
\tilde\phi_{m+1}
=
\sat_{[\phi_{\min},\phi_{\max}]}
(\phi_m+\Delta\phi_m),
\quad
\phi_{m+1}
=
(1-\lambda)\phi_m+\lambda\tilde\phi_{m+1}.
\]

\State Extract $\phi_{m+1}=[K_p,K_i,K_d]^T$.

\If{$\sign(e_m)\sign(e_{m-1})<0$}
    \State Apply the sign-flip guard:
    \[
    I_m\leftarrow 0.5I_m,
    \qquad
    K_p\leftarrow \gamma_{\rm flip}K_p,
    \qquad 0<\gamma_{\rm flip}<1 .
    \]
\EndIf

\State Set
\[
K_i^{\rm eff}
=
\begin{cases}
1.25K_i, & |e_m|<e_{\rm small},\\
K_i, & |e_m|\ge e_{\rm small}.
\end{cases}
\]

\State Update the integral state:
\[
I_m
\leftarrow
\sat_{[-I_{\max},I_{\max}]}
\!\left((1-I_{\rm leak})I_m+e_mT_s\right).
\]

\State Compute the unsaturated PID command:
\[
EF_m^{\rm unsat}
=
K_pe_m+K_i^{\rm eff}I_m-K_dd_m^{(f)} .
\]

\State Apply EF rate limiting and saturation:
\[
EF_{m+1}
=
\sat_{[EF_{\min},EF_{\max}]}
\!\left(
\rate_{\Delta EF_{\max}}(EF_m^{\rm unsat},EF_m)
\right).
\]

\State Apply anti-windup:
\[
I_m
\leftarrow
I_m+
K_{b,\rm aw}
\bigl(EF_{m+1}-EF_m^{\rm unsat}\bigr)T_s .
\]

\State Map, rate-limit, and saturate the amplitude:
\[
\tilde A_{m+1}=A(EF_{m+1}), \ 
A_{m+1}
\leftarrow
\sat_{[A_{\min},A_{\max}]}
\!\left(
\rate_{\Delta A_{\max}}(\tilde A_{m+1},A_m)
\right).
\]

\If{$\texttt{first\_band\_hit}=1$}
    \State Enforce the post-band soft floor:
    \[
    A_{m+1}\leftarrow \max\{A_{m+1},0.98A_{\mathrm{basal}}\}.
    \]
\EndIf

\State \Return $(\phi_{m+1},I_m,EF_{m+1},A_{m+1})$.
\end{algorithmic}
\end{algorithm}
Together, Algorithms~\ref{alg:window_cost}-\ref{alg:normal_adaptive_PID} separate the controller implementation into three reusable components: predictive cost evaluation, basal holding after target entry, and normal adaptive PID updating. This separation clarifies that the live controller remains a sampled-data PID law with practical safeguards, while model-assisted rollouts are used only to compute the gain-update direction. The notation also makes explicit where actuator constraints, anti-windup, amplitude mapping, and band-lock supervision enter the closed-loop simulation.
\subsection{Local sampled-data input-to-state Stability (ISS) interpretation of APID}
\label{subsec:apid_iss}

We next give a compact local stability interpretation of the APID closed loop at the
sampling instants \(t_m=t_0+mT_s\). The result follows the standard
discrete-time input-to-state stability (ISS) Lyapunov argument
\cite{JiangWang2001,JiangWang2002}. Its role here is to interpret the APID
closed loop as a practical sampled-data regulator: model mismatch, measurement
noise, burst-averaging error, and implementation error enter as bounded
disturbance inputs. This viewpoint is consistent with emulation-based
sampled-data nonlinear control and sampled-data practical output regulation
\cite{NesicTeel2004,NesicTeelCarnevale2009,LiuHuang2021}.

Let the delayed EF-driven plant be written as
\[
\dot y(t)=F(y(t),EF(t),d(t)),
\qquad z(t)=h(y(t))=y_4(t),
\]
where \(d(t)\) collects model mismatch and unmodeled disturbances. Let \(y^\star\) satisfy \(h(y^\star)=T_4^{\rm des}\), and write the sampled controller-side output as
\[
z_m=h(y(t_m))+\eta_m.
\]
Define
\[
\xi_m=\left[(y(t_m)-y^\star)^\top,\ I_m-I^\star,\ d_m^{(f)}-d_\star^{(f)},\ (\phi_m-\phi^\star)^\top\right]^\top
\]
and
\[
w_m=\left[\bar d_m,\eta_m,{(\eta_m-\eta_{m-1})}/{T_s},\delta_m\right]^\top,
\]
where \(\bar d_m\) is sampled model mismatch and \(\delta_m\) collects burst-averaging, finite-difference, numerical, and implementation errors.

\begin{theorem}
\label{thm:apid_iss}
Consider the APID closed loop with EF saturation, EF rate limits, amplitude
bounds, gain projection, and the band-lock logic described above. 
Suppose that, on a compact operating set \(\mathcal D\), the sampled closed-loop dynamics are well posed and the output map \(h\) is locally Lipschitz on the
projection of \(\mathcal D\) onto the plant-state coordinates, and there exist a function \(V:\mathcal D\to\mathbb R_{\ge0}\) and
constants \(a_1,a_2,a_3,a_4>0\) such that
\[
a_1\|\xi_m\|^2
\le
V(\xi_m)
\le
a_2\|\xi_m\|^2,
\]
and
\[
V(\xi_{m+1})-V(\xi_m)
\le
-a_3\|\xi_m\|^2+a_4\|w_m\|^2
\]
for all sampled trajectories in \(\mathcal D\). Then the sampled APID closed
loop is locally input-to-state stable with respect to \(w_m\). More precisely,
let
\[
\bar a_3=\min\{a_3,a_2/2\},
\qquad
q=1-\frac{\bar a_3}{a_2}.
\]
Then \(q\in(0,1)\), and
\[
\|\xi_m\|
\le
\sqrt{\frac{a_2}{a_1}}\,q^{m/2}\|\xi_0\|
+
\sqrt{\frac{a_4a_2}{a_1\bar a_3}}\,
\sup_{0\le j\le m-1}\|w_j\|.
\]
Consequently, if $\bar w:=\sup_{j\ge0}\|w_j\|<\infty$,
then \(\xi_m\) is uniformly ultimately bounded and
\[
\limsup_{m\to\infty}\|\xi_m\|
\le
\sqrt{\frac{a_4a_2}{a_1\bar a_3}}\,\bar w .
\]
Moreover, there exists \(c_e>0\) such that
\[
\limsup_{m\to\infty}|e_m|
\le
c_e \bar w .
\]
In particular, if \(w_m\equiv0\), then the disturbance-free sampled APID equilibrium is exponentially stable.
\end{theorem}

\begin{proof}
Since \(\bar a_3=\min\{a_3,a_2/2\}\), the Lyapunov difference inequality implies
\[
V(\xi_{m+1})-V(\xi_m)
\le
-\bar a_3\|\xi_m\|^2+a_4\|w_m\|^2 .
\]
Using \(V(\xi_m)\le a_2\|\xi_m\|^2\), we have
\[
\|\xi_m\|^2\ge \frac{1}{a_2}V(\xi_m).
\]
Therefore,
\[
V(\xi_{m+1})
\le
\left(1-\frac{\bar a_3}{a_2}\right)V(\xi_m)
+
a_4\|w_m\|^2
=
qV(\xi_m)+a_4\|w_m\|^2 .
\]
Because \(0<\bar a_3\le a_2/2\), \(q\in[1/2,1)\). Iterating the preceding
inequality gives
\[
V(\xi_m)
\le
q^m V(\xi_0)
+
a_4\sum_{j=0}^{m-1}q^{m-1-j}\|w_j\|^2 .
\]
Since
\[
\sum_{j=0}^{m-1}q^{m-1-j}\le \frac{1}{1-q}
=\frac{a_2}{\bar a_3},
\]
we obtain
\[
V(\xi_m)
\le
q^m V(\xi_0)
+
\frac{a_4a_2}{\bar a_3}
\sup_{0\le j\le m-1}\|w_j\|^2 .
\]
If the disturbance sequence $w_m$ is uniformly bounded, i.e.,
$\bar w:=\sup_{j\ge0}\|w_j\|<\infty$,
then the preceding estimate implies
\[
V(\xi_m)
\le
q^m V(\xi_0)
+
\frac{a_4a_2}{\bar a_3}\bar w^2 .
\]
Using the quadratic bounds on \(V\), 
\[
a_1\|\xi_m\|^2\le V(\xi_m),
\qquad
V(\xi_0)\le a_2\|\xi_0\|^2,
\]
we get
\[
\|\xi_m\|
\le
\sqrt{\frac{a_2}{a_1}}\,q^{m/2}\|\xi_0\|
+
\sqrt{\frac{a_4a_2}{a_1\bar a_3}}\,
\sup_{0\le j\le m-1}\|w_j\|. 
\]
Taking the limit superior as \(m\to\infty\), the exponentially decaying term
vanishes, giving the stated ultimate bound for \(\xi_m\).

It remains to relate the state bound to the sampled tracking error. Since
\[
e_m=T_4^{\rm des}-z_m
=
h(y^\star)-h(y(t_m))-\eta_m,
\]
by the local Lipschitz assumption on \(h\), there exists \(L_h>0\) on the
plant-state projection of \(\mathcal D\) such that

\[
|e_m|
\le
L_h\|y(t_m)-y^\star\|+|\eta_m|
\]
for some \(L_h>0\). Because \(y(t_m)-y^\star\) is a component of \(\xi_m\) and
\(\eta_m\) is a component of \(w_m\), the ISS estimate implies
\[
\limsup_{m\to\infty}|e_m|
\le
c_e\sup_{j\ge0}\|w_j\|
\]
for a constant \(c_e>0\). If \(w_m\equiv0\), then
\(V(\xi_{m+1})\le qV(\xi_m)\), and the same quadratic bounds imply exponential
convergence of \(\xi_m\) to zero. This completes the proof.
\end{proof}

Theorem~\ref{thm:apid_iss} should be read as a practical stability statement.
It does not claim global asymptotic convergence of the full delayed biological
plant. Rather, it states that, whenever the sampled APID closed loop admits the
one-step ISS Lyapunov inequality above, the measured tracking error remains
inside a disturbance-dependent neighborhood of the target . In the disturbance-free case, this neighbourhood
collapses to the equilibrium.
\subsection{Computational closed-loop validation for \texorpdfstring{$T_4$}{T4} regulation}
We evaluate APID on the \(16\)-state model with the burst-averaged actuator. The numerical studies test target tracking across setpoints, bounded actuation under hardware constraints, retention near the target after band entry, and robustness under uncertainty.

\begin{figure}[t]
\centering
\includegraphics[width=\linewidth]{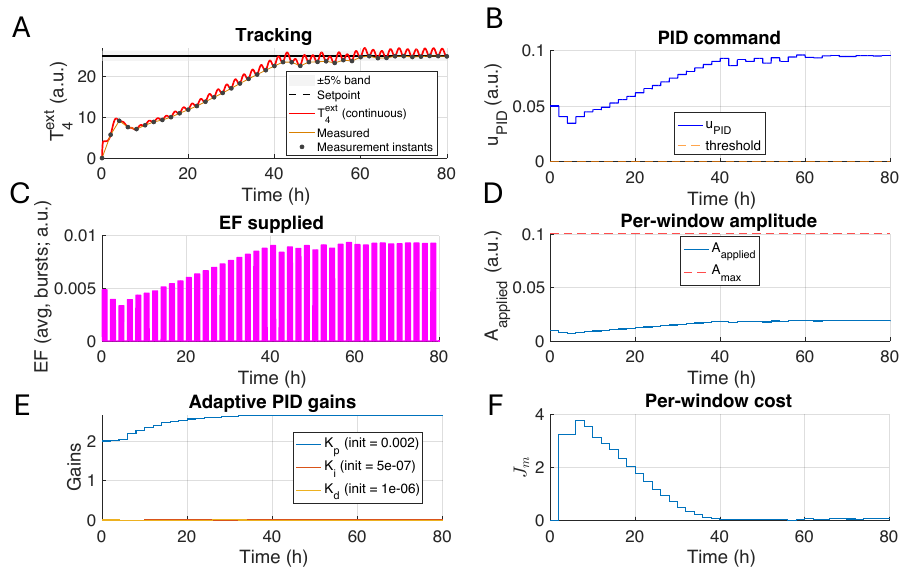}
\caption{Closed-loop performance of APID with band-lock holding for \(T_4^{\rm des}=25\) (a.u.). Panels show tracking, PID command, burst-averaged EF, applied amplitude, adaptive gains, and one-window predictive cost.}
\label{fig:apid_bandlock_setpoint25}
\end{figure}

\Cref{fig:apid_bandlock_setpoint25} shows a representative trajectory for \(T_4^{\rm des}=25\) (a.u.). The output rises toward the target, the sampled measurements remain close to the continuous trajectory, and the EF command and applied amplitude evolve gradually under the imposed rate limits and bounds. In the simulations, \(\mathrm{thr}=10^{-4}\), \(A_{\min}=0\), \(A_{\max}=0.10\), \(k_A=0.20\), \(EF_{\min}=0\), \(EF_{\max}=100\), and \(\Delta EF_{\max}=25\), all in arbitrary units. The dominant adaptation occurs in \(K_p\), while the predictive cost decreases over the transient. After target entry, band-lock maintains regulation near the setpoint.

\begin{figure}[t]
\centering
\includegraphics[width=\linewidth]{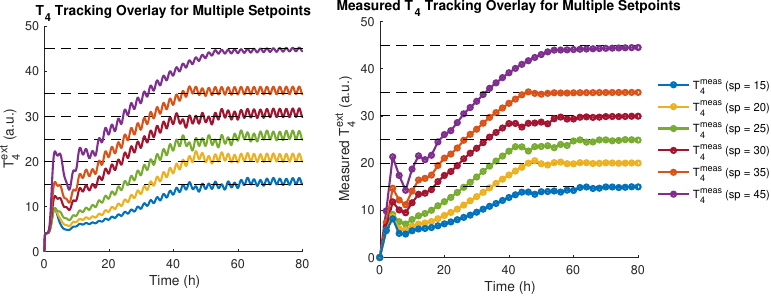}
\caption{Closed-loop trajectories for multiple setpoints \(T_4^{\rm des}\in\{15,20,25,30,35,45\}\) under APID with band-lock holding.}
\label{fig:T4_overlay_multiple_setpoints}
\end{figure}

\Cref{fig:T4_overlay_multiple_setpoints} shows that the same controller structure regulates across several target levels. Larger setpoints lead to larger steady-state hormone levels, and all trajectories show the same qualitative pattern: an initial rise followed by convergence toward the desired neighborhood with small bounded oscillations. The measured-data sweep yields small final errors and band-entry times close to settling times; for example, at setpoints \(20,25,30,35,45\), the rise times are approximately \(50,38,36,36,42\) h, respectively, and the final measured errors are small relative to the target.

A fixed-gain sampled-data PID baseline under the same actuator constraints fails to reach the desired level and saturates both the EF command and applied amplitude, whereas APID reaches the target neighborhood using substantially smaller actuation. Similarly, comparison with an APID variant without band-lock shows that the band-lock mechanism improves near-setpoint retention and reduces late-time deviation. The fixed-gain PID and no-band-lock comparisons are summarized in \Cref{fig:fixedpid_vs_apid}, \Cref{fig:bandlock_vs_nobandlock}, and  \Cref{tab:bandlock_vs_nobandlock}.

\begin{table}[t]
\centering
\normalsize
\setlength{\tabcolsep}{6pt}
\renewcommand{\arraystretch}{1.05}
\caption{Measured-data performance for the multi-setpoint APID sweep.}
\label{tab:measured_T4_metrics}
\begin{tabular}{c c c c c c}
\hline
\(T_4^{\mathrm{des}}\) & Rise (h) & Peak & OS (\%) & Final err. & Settling (h) \\
\hline
20 & 50 & 20.594 & 2.97 &  0.072 & 54 \\
25 & 38 & 24.973 & 0    & -0.074 & 58 \\
30 & 36 & 29.930 & 0    & -0.070 & 54 \\
35 & 36 & 35.112 & 0.32 & -0.021 & 42 \\
45 & 42 & 44.514 & 0    & -0.486 & 50 \\
\hline
\end{tabular}
\end{table}

\Cref{tab:measured_T4_metrics} summarizes the measured-data APID sweep.
Here the rise time is \(t_{90}-t_{10}\), overshoot is
\[
\%OS=100\,\frac{\max(0,z_{\max}-T_4^{\mathrm{des}})}{T_4^{\mathrm{des}}},
\]
and settling time is the first sampling time after which all subsequent samples remain within the \(\pm5\%\) target band. Across all tested setpoints, APID reaches the target neighborhood with small overshoot and bounded final error.

\begin{figure}[t]
    \centering
    \includegraphics[width=0.8\textwidth]{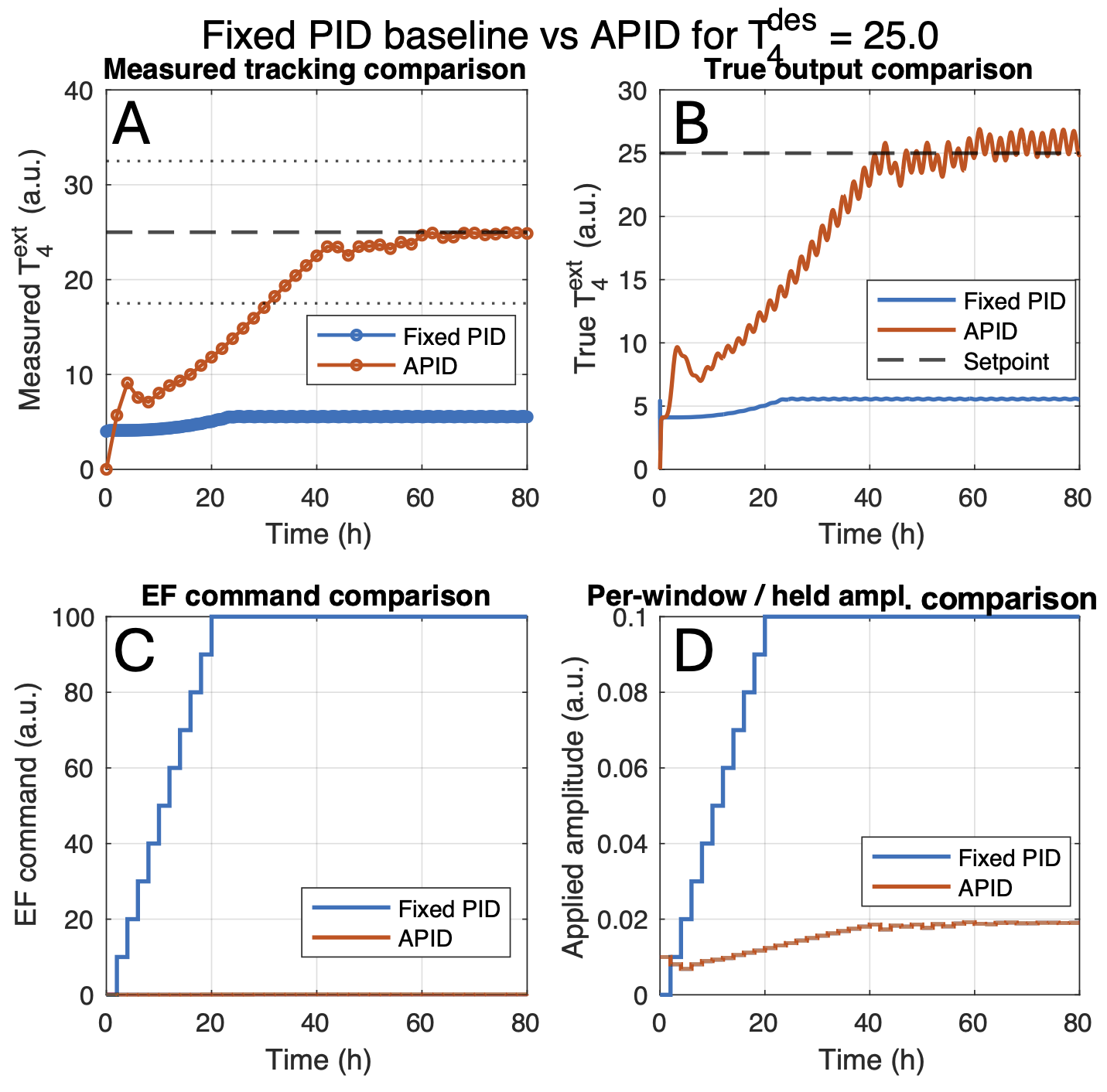}
    \caption{Comparison of a sampled-data PID baseline with fixed gains and the adaptive PID (APID) controller for \(T_4^{\mathrm{des}}=25\) (a.u.). A: measured-output tracking. B: true output trajectory. C: EF command. D: applied amplitude. The fixed-gain controller fails to drive the output toward the desired level and rapidly saturates both the EF command and the stimulation amplitude, whereas APID reaches the target neighborhood while using substantially smaller actuation.}
    \label{fig:fixedpid_vs_apid}
\end{figure}
\Cref{fig:fixedpid_vs_apid} clarifies the role of adaptation in the present delayed nonlinear setting. Under the same sampled-data architecture and actuator constraints, the fixed-gain PID baseline fails to regulate the hormone output to the desired level and instead saturates both the EF command and the applied amplitude. By contrast, APID reaches the target neighborhood with substantially smaller actuation. This comparison shows that, for the present EF-driven endocrine model, online gain adaptation is not a minor refinement but an essential component of effective regulation.

\begin{figure}[t]
    \centering
    \includegraphics[width=0.5\textwidth]{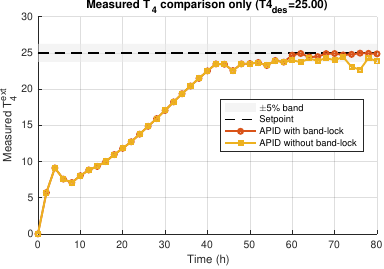}
    \caption{Comparison of measured \(T_4^{\mathrm{ext}}\) trajectories for adaptive PID control with and without the band-lock mechanism at \(T_4^{\mathrm{des}}=25\) (a.u.). The shaded region denotes the \(\pm5\%\) tolerance band around the target. The band-lock controller exhibits improved near-setpoint retention, spends more time inside the tolerance band, and achieves a final measured \(T_4\) value closer to the target.}
    \label{fig:bandlock_vs_nobandlock}
\end{figure}

\noindent \Cref{fig:bandlock_vs_nobandlock} isolates the effect of the supervisory band-lock mechanism within the adaptive PID architecture. Both controllers initially approach the target region, but the version without band-lock exhibits a larger late-time deviation below the setpoint. By contrast, once target entry is detected, the band-lock controller maintains a sustaining basal stimulation level and thereby improves retention inside the desired neighborhood. This comparison shows that band-locking is not merely a cosmetic safeguard; in the present delayed endocrine system it materially improves setpoint maintenance after convergence.

\begin{table}[t]
\centering
\normalsize
\setlength{\tabcolsep}{5pt}
\renewcommand{\arraystretch}{1.05}
\caption{Effect of band-lock on APID performance at \(T_4^{\mathrm{des}}=25\).}
\label{tab:bandlock_vs_nobandlock}
\begin{tabular}{l c c c}
\hline
Controller & IAE & Time-in-5\% & Final \(T_4^{\mathrm{ext}}\) \\
\hline
With band-lock    & 29807.7 & 31.71\% & 24.9255 \\
Without band-lock & 31228.5 & 21.95\% & 23.5787 \\
\hline
\end{tabular}
\end{table}
\noindent \Cref{tab:bandlock_vs_nobandlock} shows that band-lock improves near-target retention: it lowers integral absolute error (IAE), increases the fraction of samples inside the \(\pm5\%\) band, and gives a final
measured output closer to the setpoint.

\section{Robust Adaptive PID (RAPID) Control Design}
\label{sec:rapid}
Biological regulation in practice is subject to parameter uncertainty, sensor noise and bias, actuator mismatch, timing uncertainty, and unmodeled disturbances. RAPID preserves the APID sampled-data structure but modifies the measurement channel and gain-update mechanism. At each control window, the controller forms the tracking error from a filtered and bias-corrected measurement, updates the PID terms, and computes the next EF command subject to rate limits, saturation, and anti-windup. The robustification enters through the gain update: instead of tuning \((K_p,K_i,K_d)\) using one nominal prediction, RAPID evaluates perturbed scenarios and minimizes a mean-CVaR objective.

The measured output is modeled as
\[
z_m=T_4^{\rm ext}(t_m)+b_\eta+\eta_m,
\qquad \eta_m\sim\mathcal N(0,\sigma_\eta^2),
\]
where \(\sigma_\eta\) is the measurement-noise standard deviation and \(b_\eta\) is a constant sensor bias. The robust simulations include multiplicative uncertainty in selected cascade and thyroid-core parameters, additive exogenous disturbance in the extracellular \(T_4\) equation, actuator-gain mismatch, EF delay, jitter, and measurement corruption; see \Cref{tab:uncertainty_compact}.

Let \(d_4(t)\) denote the additive exogenous disturbance entering the extracellular
\(T_4\) dynamics:
\[
d_4(t)=b_{d_4}+a_{d_4}\sin\!\left(\frac{2\pi t}{P_{d_4}}+\phi_{d_4}\right),
\]
where \(b_{d_4}\), \(a_{d_4}\), \(P_{d_4}\), and \(\phi_{d_4}\) denote the disturbance
bias, amplitude, period, and phase, respectively.

The actuator uncertainty is represented by a multiplicative gain mismatch and
a timing perturbation. If \(EF_{\rm cmd}(t)\) denotes the burst-averaged command
generated by the controller, the perturbed actuator applies
\[
EF_{\rm app}(t)
=
(1+g_A)\,EF_{\rm cmd}\!\left(t-\tau_A-\nu_{A,m}\right),
t\in[t_m,t_{m+1}),
\]
where \(g_A\) is a dimensionless actuator-gain error, \(\tau_A\ge0\) is a persistent
delay, and \(\nu_{A,m}\) is a windowwise jitter satisfying
\[
|\nu_{A,m}|\le \sigma_A .
\]
Thus, \(\sigma_A\) is used as a jitter bound, not as measurement noise.
\begin{table}[t]
\caption{Uncertainty classes used in RAPID simulations.}
\label{tab:uncertainty_compact}
\centering
\small
\setlength{\tabcolsep}{6pt}
\renewcommand{\arraystretch}{1.55}
\begin{tabular}{p{0.24\linewidth}p{0.68\linewidth}}
\toprule
Class & Distribution or range \\
\midrule
EF cascade & \(\alpha_1,K,\gamma_1,\alpha\sim\mathcal U(\pm 8\%)\) relative to nominal values \\
Thyroid core & \(\kappa,K_2\sim\mathcal U(\pm10\%)\); \(\gamma_I,\alpha_{T_g},\gamma_{T_g},\gamma_{T_4}^{\rm int},\omega,\gamma_{T_4}^{\rm ext}\sim\mathcal U(\pm8\%)\) \\
Disturbance & \(d_4(t)=b_{d_4}+a_{d_4}\sin(2\pi t/P_{d_4}+\phi_{d_4})\), with \(b_{d_4}\sim\mathcal U(\pm0.01)\), \(a_{d_4}\sim\mathcal U(0,0.01)\), \(P_{d_4}=360\) min \\
Actuator & \(EF_{\rm app}(t)=(1+g_A)EF_{\rm cmd}(t-\tau_A-\nu_{A,m})\), with \(g_A\sim\mathcal U(\pm10\%)\), \(\tau_A\sim\mathcal U(0,5)\) s, \(|\nu_{A,m}|\le\sigma_A\), \(\sigma_A\sim\mathcal U(0,2)\) s \\
Measurement & \(\sigma_\eta=0.8\), \(b_\eta=2.0\) \\
\bottomrule
\end{tabular}
\end{table}
The complete RAPID implementation used in the robust closed-loop simulations is summarized in \Cref{alg:rapid-robust}.
\begin{algorithm}[H]
\footnotesize
\caption{RAPID: robust sampled-data adaptive PID with lock mode and scenario-CVaR gain tuning}
\label{alg:rapid-robust}
\begin{algorithmic}[1]
\Require Setpoint \(T_4^{\rm des}\), window length \(T_s\), initial gains, actuator constraints, robust settings \((M,\lambda,q)\), and random seed.
\State Sample one persistent true perturbation and initialize measurement filter, bias estimator, PID states, and band-lock variables.
\For{each control window \(m\)}
\State Apply the held burst-averaged EF waveform with actuator gain mismatch, delay, and jitter; propagate the perturbed plant over \([t_m,t_m+T_s)\).
\State Form raw, filtered, and bias-corrected measurements; update derivative-on-measurement and tracking error.
\State Update band-entry/band-exit logic and the basal amplitude if lock is active.
\State Sample \(M\) scenario perturbations and evaluate \(J_m^{\rm rob}\) in \eqref{eq:rapid_cost}.
\If{locked}
\State Update \(A_{m+1}\) using the basal holding law and map it to \(EF_{m+1}\).
\Else
\State Approximate \(\nabla_\phi J_m^{\rm rob}\) by bounded central differences; project and smooth the gain update.
\State Compute the PID command, apply EF saturation/rate limits, anti-windup, amplitude mapping, amplitude limits, and the post-band sustaining floor.
\EndIf
\State Log true output, used measurement, EF command, amplitude, gains, lock state, and robust cost.
\EndFor
\end{algorithmic}
\end{algorithm}

\subsection{Local risk-aware sampled-data ISS interpretation of RAPID}
\label{subsec:rapid_stability}
Let the perturbed plant be written as
\[
\dot y(t)=F(y(t),EF_{\rm app}(t),d(t)),
\qquad z(t)=h(y(t))=y_4(t).
\]
At \(t_m\), RAPID uses a filtered and bias-corrected measurement \(\widehat z_m\) and filtered reference \(r_{f,m}\), so that \(e_m=r_{f,m}-\widehat z_m\). The robust adaptive gain vector \(\phi_m=[K_p^{(m)},K_i^{(m)},K_d^{(m)}]^\top\) is updated using
\begin{equation}
\label{eq:rapid_cost}
J_m^{\rm rob}(\phi)=
(1-\lambda)\frac1M\sum_{i=1}^M J_m^{(i)}(\phi)
+\lambda\,\mathrm{CVaR}_q\big(\{J_m^{(i)}(\phi)\}_{i=1}^M\big),
\end{equation}
where \(J_m^{(i)}\) is the one-window cost under scenario \(i\), \(\lambda\in[0,1]\), and \(q\in(0,1)\). The projected gain update is
\[
\phi_{m+1}=\Pi_\Phi(\phi_m-\alpha_m\widehat g_m),
\qquad \widehat g_m=g_m+r_m,
\]
where \(r_m\) collects finite-difference, finite-scenario, numerical-prediction, and CVaR-approximation errors.

Let \(y^\star\) satisfy \(h(y^\star)=T_4^{\rm des}\), define \(\widetilde\phi_m=\phi_m-\phi^\star\), and let \(\xi_m\) collect the plant deviation, integral-state deviation, derivative-filter deviation, bias-estimator deviation, and \(\widetilde\phi_m\). Let \(w_m\) collect plant mismatch, residual measurement/filtering error, derivative-noise increments, actuator-gain error, delay, jitter, burst-averaging error, and \(r_m\).

\begin{theorem}
\label{thm:rapid_risk_iss}
Consider the RAPID closed loop with EF saturation, EF rate limits, amplitude
bounds, band-locking, and projected gain updates. Suppose that, on a compact
operating set \(\mathcal D\), the sampled closed-loop dynamics are well posed
and the output map \(h\) is locally Lipschitz on the plant-state projection of
\(\mathcal D\). Assume the following.

\begin{enumerate}
\item There exists a plant-regulation Lyapunov function \(V_y\) and constants
\(b_1,b_2,b_3,b_4,b_5>0\) such that
\[
b_1\|\chi_m\|^2
\le
V_y(\chi_m)
\le
b_2\|\chi_m\|^2,
\]
and
\[
V_y(\chi_{m+1})-V_y(\chi_m)
\le
-b_3\|\chi_m\|^2
+
b_4\|\widetilde\phi_m\|^2
+
b_5\|w_m\|^2 .
\]

\item The local risk-gradient component satisfies, uniformly in \(m\),
\[
\langle g_m,\widetilde\phi_m\rangle
\ge
\mu\|\widetilde\phi_m\|^2,
\qquad
\|g_m\|
\le
L\|\widetilde\phi_m\|,
\]
for some \(\mu,L>0\).

\item The step size satisfies
\[
0<\underline\alpha\le \alpha_m\le \bar\alpha<\frac{2\mu}{L^2}.
\]
\end{enumerate}

Then the sampled RAPID closed loop is locally input-to-state stable with
respect to \(w_m\). In particular, there exist constants \(C_1,C_2>0\) and
\(\rho\in(0,1)\) such that
\[
\|\xi_m\|
\le
C_1\rho^m\|\xi_0\|
+
C_2\sup_{0\le j\le m-1}\|w_j\|.
\]
If
\[
\bar w:=\sup_{j\ge0}\|w_j\|<\infty,
\]
then
\[
\limsup_{m\to\infty}\|\xi_m\|
\le
C_2\bar w .
\]
Moreover, there exists \(C_e>0\) such that
\[
\limsup_{m\to\infty}|e_m|
\le
C_e\bar w .
\]
If \(w_m\equiv0\), then the disturbance-free sampled RAPID equilibrium is
exponentially stable.
\end{theorem}

\begin{proof}
The plant-regulation assumption gives
\[
V_y(\chi_{m+1})-V_y(\chi_m)
\le
-b_3\|\chi_m\|^2
+
b_4\|\widetilde\phi_m\|^2
+
b_5\|w_m\|^2 .
\]
The term \(b_4\|\widetilde\phi_m\|^2\) accounts for the fact that imperfect PID
gains perturb the plant-regulation dynamics.

Next consider the gain update. Since Euclidean projection onto a closed convex
set is nonexpansive and \(\phi^\star\in\Phi\),
\[
\|\widetilde\phi_{m+1}\|^2
\le
\left\|
\widetilde\phi_m-\alpha_m(g_m+r_m)
\right\|^2 .
\]
Using Young's inequality, for any \(\varepsilon>0\),
\[
\|\widetilde\phi_{m+1}\|^2
\le
(1+\varepsilon)
\left\|
\widetilde\phi_m-\alpha_m g_m
\right\|^2
+
\left(1+\frac1\varepsilon\right)\alpha_m^2\|r_m\|^2 .
\]
By the risk-gradient conditions,
\[
\left\|
\widetilde\phi_m-\alpha_m g_m
\right\|^2
\le
\left(1-2\alpha_m\mu+\alpha_m^2L^2\right)
\|\widetilde\phi_m\|^2 .
\]
Because
\[
\underline\alpha\le \alpha_m\le \bar\alpha<\frac{2\mu}{L^2},
\]
there exists \(\sigma_\phi\in(0,1)\), after choosing \(\varepsilon>0\) small
enough, such that
\[
\|\widetilde\phi_{m+1}\|^2-\|\widetilde\phi_m\|^2
\le
-\sigma_\phi\|\widetilde\phi_m\|^2
+
c_r\|r_m\|^2
\]
for some \(c_r>0\). Since \(r_m\) is a component of \(w_m\),
\[
\|\widetilde\phi_{m+1}\|^2-\|\widetilde\phi_m\|^2
\le
-\sigma_\phi\|\widetilde\phi_m\|^2
+
c_r\|w_m\|^2 .
\]

Define the composite Lyapunov function
\[
V(\xi_m)=V_y(\chi_m)+\kappa\|\widetilde\phi_m\|^2 .
\]
Combining the two inequalities gives
\begin{eqnarray*}
V(\xi_{m+1})-V(\xi_m)
&\le
-b_3\|\chi_m\|^2
-
(\kappa\sigma_\phi-b_4)\|\widetilde\phi_m\|^2 \\
&+
(b_5+\kappa c_r)\|w_m\|^2 .
\end{eqnarray*}
Choose
\[
\kappa>\frac{b_4}{\sigma_\phi}.
\]
Then there exist constants \(c_1,c_2>0\) such that
\[
V(\xi_{m+1})-V(\xi_m)
\le
-c_1\|\xi_m\|^2+c_2\|w_m\|^2 .
\]
Also, by the quadratic bounds on \(V_y\), there exist \(c_3,c_4>0\) such that
\[
c_3\|\xi_m\|^2
\le
V(\xi_m)
\le
c_4\|\xi_m\|^2 .
\]
Let
\[
\bar c_1=\min\{c_1,c_4/2\},
\qquad
q=1-\frac{\bar c_1}{c_4}\in(0,1).
\]
Then
\[
V(\xi_{m+1})
\le
qV(\xi_m)+c_2\|w_m\|^2 .
\]
Iterating gives
\[
V(\xi_m)
\le
q^mV(\xi_0)
+
c_2\sum_{j=0}^{m-1}q^{m-1-j}\|w_j\|^2 .
\]
Since
\[
\sum_{j=0}^{m-1}q^{m-1-j}\le \frac{1}{1-q},
\]
the quadratic bounds on \(V\) imply that there exist \(C_1,C_2>0\) and
\(\rho=\sqrt q\in(0,1)\) such that
\[
\|\xi_m\|
\le
C_1\rho^m\|\xi_0\|
+
C_2\sup_{0\le j\le m-1}\|w_j\|.
\]
If \(\bar w=\sup_{j\ge0}\|w_j\|<\infty\), taking the limit superior yields
\[
\limsup_{m\to\infty}\|\xi_m\|
\le
C_2\bar w .
\]

It remains to bound the sampled tracking error. Since
\[
e_m=r_{f,m}-\widehat z_m,
\]
write
\[
e_m=
\bigl(r_{f,m}-T_4^{\rm des}\bigr)
+
\bigl(h(y^\star)-h(y(t_m))\bigr)
-
\widetilde\eta_m .
\]
The filtered-reference deviation \(r_{f,m}-T_4^{\rm des}\) is included in the
sampled perturbation channel, or is zero after the reference transient. By the
local Lipschitz assumption on \(h\),
\[
|h(y^\star)-h(y(t_m))|
\le
L_h\|y(t_m)-y^\star\|.
\]
Since \(y(t_m)-y^\star\) is a component of \(\xi_m\) and
\(\widetilde\eta_m\) is a component of \(w_m\), the ISS estimate gives
\[
\limsup_{m\to\infty}|e_m|
\le
C_e\bar w
\]
for some \(C_e>0\). If \(w_m\equiv0\), then
\(V(\xi_{m+1})\le qV(\xi_m)\), which gives exponential convergence of
\(\xi_m\) to zero. This completes the proof.
\end{proof}
The assumptions involving \(g_m\) are the part that distinguishes RAPID from
APID. They encode local descent of the mean-CVaR scenario cost, while the
residual \(r_m\) captures the finite-scenario, finite-difference, and numerical
errors introduced by the implementable robust update.
\subsection{Robust Closed-Loop Validation for \texorpdfstring{$T_4$}{T4} Regulation}
\begin{figure}[t]
\centering
\includegraphics[width=\linewidth]{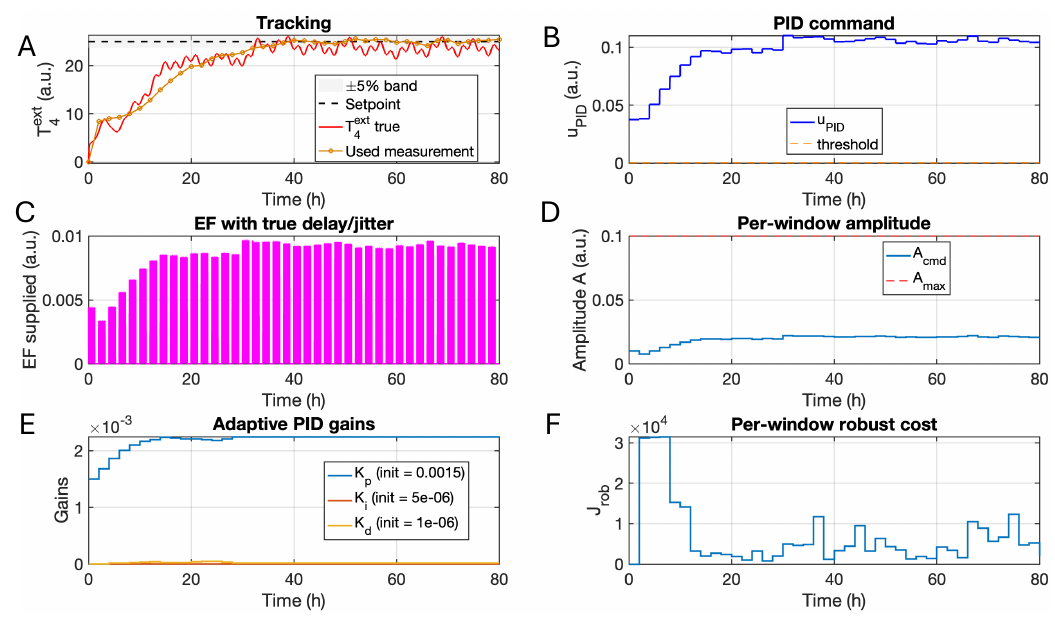}
\caption{Closed-loop performance of RAPID under five simultaneous perturbation classes for \(T_4^{\rm des}=25\) (a.u.). Panels show true and used-measurement tracking, PID command, perturbed EF supplied to the plant, commanded amplitude, adaptive gains, and robust cost.}
\label{fig:rapid_five_perturbations}
\end{figure}

\Cref{fig:rapid_five_perturbations} shows that RAPID achieves stable long-horizon tracking under simultaneous perturbations. The used measurement is smoother than the true trajectory and provides the controller-side signal for adaptation. The gain trajectories remain well behaved, with \(K_p\) showing the dominant adaptation while \(K_i\) and \(K_d\) remain comparatively small. The robust cost decreases overall, despite window-to-window fluctuations, suggesting that the adaptive law learns a safer control policy across perturbed scenarios.

\begin{remark}
In RAPID, the controller-side tracking error is formed using the filtered reference \(r_f(t)\), rather than the constant target directly. Consequently, derivative-on-measurement is no longer exactly equivalent to derivative-on-error, since \(\dot e(t)=\dot r_f(t)-\dot y^{\rm used}(t)\). The derivative term is nevertheless taken on the measurement channel to avoid derivative kick and reduce sensitivity to noisy, bias-corrected measurements.
\end{remark}

\begin{figure}[t]
\centering
\includegraphics[width=\linewidth]{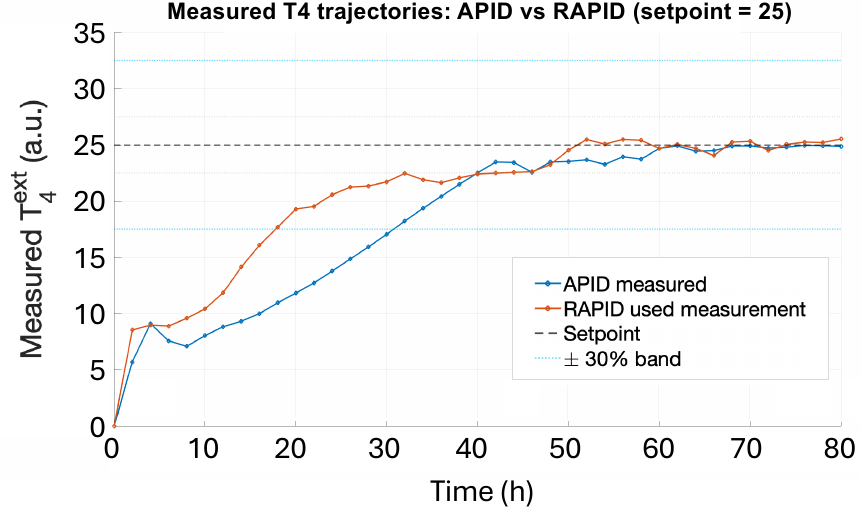}
\caption{Measured-output comparison of APID and RAPID for \(T_4^{\rm des}=25\) (a.u.).}
\label{fig:apid_vs_rapid_measured_compare}
\end{figure}

\begin{table}[H]
\caption{Measured-data comparison of APID and RAPID for \(T_4^{\rm des}=25\) (a.u.).}
\label{tab:apid-vs-rapid}
\centering
\normalsize
\setlength{\tabcolsep}{6pt}
\renewcommand{\arraystretch}{1.05}
\begin{tabular}{lcccccc}
\toprule
Controller & OS (\%) & Rise (h) & Settle (h) & In 30\% (\%) & In 10\% (\%) & IAE \\
\midrule
APID & 0.00 & 40 & 32 & 60.98 & 51.22 & 527.896 \\
RAPID & 2.10 & 44 & 18 & 78.05 & 46.34 & 373.265 \\
\bottomrule
\end{tabular}
\end{table}

\Cref{fig:apid_vs_rapid_measured_compare} and \Cref{tab:apid-vs-rapid} compare APID and RAPID for the same target value. APID reaches the target with no overshoot and achieves a slightly larger fraction of time inside the tighter \(\pm10\%\) band. RAPID settles faster, spends more time inside the broader \(\pm30\%\) band, and yields a substantially smaller measured-trajectory IAE, at the cost of a small overshoot.

\section{Conclusion and Discussion}
\label{sec:conclusion}
We demonstrate robust closed-loop regulation of extracellular \(T_4\) production in an electrically stimulated engineered thyroid-like cell system. Using a delayed control-oriented model and a burst-based actuator representation, APID and RAPID controllers are developed and tested for hormone regulation under sparse measurements and constrained stimulation. RAPID augments APID with filtered and bias-corrected measurements, explicit uncertainty modeling, and a scenario-based mean-CVaR gain update. In silico, this architecture regulates \(T_4\) across multiple targets and maintains tracking near the desired setpoint under simultaneous perturbations.

The results clarify the main tradeoff. Open-loop stimulation alone is insufficient for precise regulation, whereas feedback drives the system toward prescribed targets over a meaningful operating range. Although model predictive control is a natural alternative for delayed constrained systems, adaptive PID is used here for implementability: it uses measured output information, handles saturation and rate limits directly, and provides an interpretable feedback baseline for this bioelectronic application.

The study is computational, and the uncertainty model is stylized rather than learned from experiments. The ISS results should be interpreted as local sampled-data practical-stability statements: under the stated Lyapunov and bounded-disturbance assumptions, the sampled tracking error remains ultimately bounded, and the disturbance-free sampled equilibrium is exponentially stable. They do not constitute a global stability proof for the full nonlinear delayed biological plant with saturation and hysteretic switching. Future work will identify filtering, bias-correction, and uncertainty models from EF-driven cell-factory data, test transfer to in vitro platforms, and verify ISS conditions directly from identified sampled closed-loop models.

\section*{Code Availability}
The MATLAB implementation of the APID and RAPID algorithms is available at
\url{https://github.com/papridey/t4-apid-rapid-control}. The repository provides the controller code, shared plant model, burst-averaged actuator routines, and experiment scripts used to reproduce the setpoint-tracking and robustness studies.

\section*{Acknowledgments}
This research was funded, in part, by the U.S. Government through ARPA-H Agreement No. 140D042490005. The views and conclusions contained in this document are those of the authors and should not be interpreted as representing the official policies, either expressed or implied, of the United States Government. The authors thank Gary Fedder's team at Carnegie Mellon University for sharing
the timing diagram and for helpful discussions that informed the design of the actuator model. Generative AI tools were used to improve grammar and syntax in portions of the manuscript. The authors reviewed and edited all AI-assisted text and take full responsibility for the content.

\bibliographystyle{plain}
\bibliography{thyrocontr}

\section{Supplementary Materials}
This section collects additional numerical checks supporting the main results. We report the actuator time-resolution study, discuss the measurement/update interval, and provide a RAPID setpoint sweep. These studies document implementation robustness and reproducibility, but are not required for the main APID/RAPID controller construction.
\subsection{Actuator time-resolution study}
\label{subsec:time-resolution}

Because EF is delivered through short within-window bursts, the internal simulation step \(dt\) must be small enough to resolve the burst-averaged actuator signal. This numerical resolution is distinct from the control update period \(T_s\), which determines when measurements and controller updates occur. To assess sensitivity to \(dt\), we compared \(dt=0.1\,\mathrm{s}\) and \(dt=0.1/12\,\mathrm{s}\) with all controller, actuator, and model parameters fixed; see Fig.~\ref{fig:time_resolution}. The two resolutions produce closely matching measured \(T_4^{\mathrm{ext}}\) trajectories and similar EF profiles, indicating that the closed-loop behavior is not materially changed by this refinement, although \(dt\) should still be chosen fine enough to resolve burst timing.

\begin{figure}[t]
    \centering
    \includegraphics[width=\textwidth]{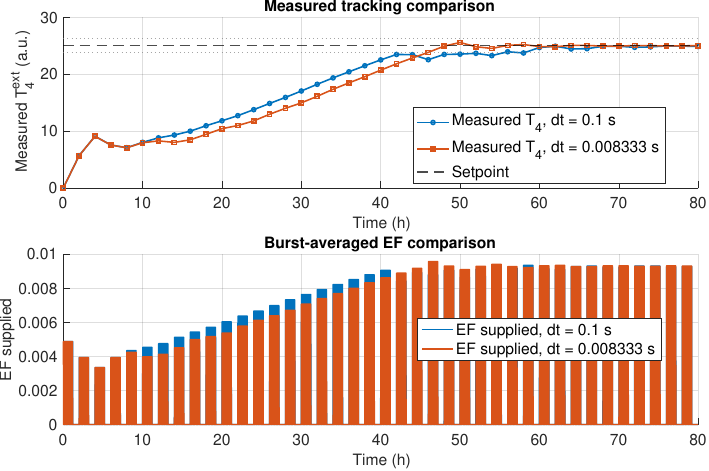}
    \caption{Time-resolution comparison for \(dt=0.1\,\mathrm{s}\) and \(dt=0.008333\,\mathrm{s}\). Top: measured \(T_4^{\mathrm{ext}}\) trajectories. Bottom: burst-averaged EF signals supplied to the plant.}
    \label{fig:time_resolution}
\end{figure}

\subsection{Effect of measurement and update interval}

The measurement/update interval is a separate sampled-data design parameter. It affects the sampled error, derivative estimate, and gain adaptation. In this delayed endocrine model, overly frequent updates can make the controller react to short-term fluctuations before the biological response to previous EF stimulation has developed. A longer update interval gives a smoother sampled signal and a more representative one-window adaptation cost. Thus, the update period should be matched to the slower biological response time scale rather than to the micro-pulse actuation scale.

\subsection{Additional closed-loop RAPID results across target setpoints}
As a supplementary validation study, we report an additional RAPID sweep over the target values \(T_4^{\mathrm{des}}\in\{15, 20,25,30,35,45\}\). The purpose of this appendix experiment is to examine how the same robust adaptive PID architecture behaves across multiple operating points under the same uncertainty model. For each setpoint, we display the used measurement, the true extracellular output, the commanded stimulation amplitude, and the per-window robust cost.
\begin{figure}[t]
\centering
\includegraphics[width=\textwidth]{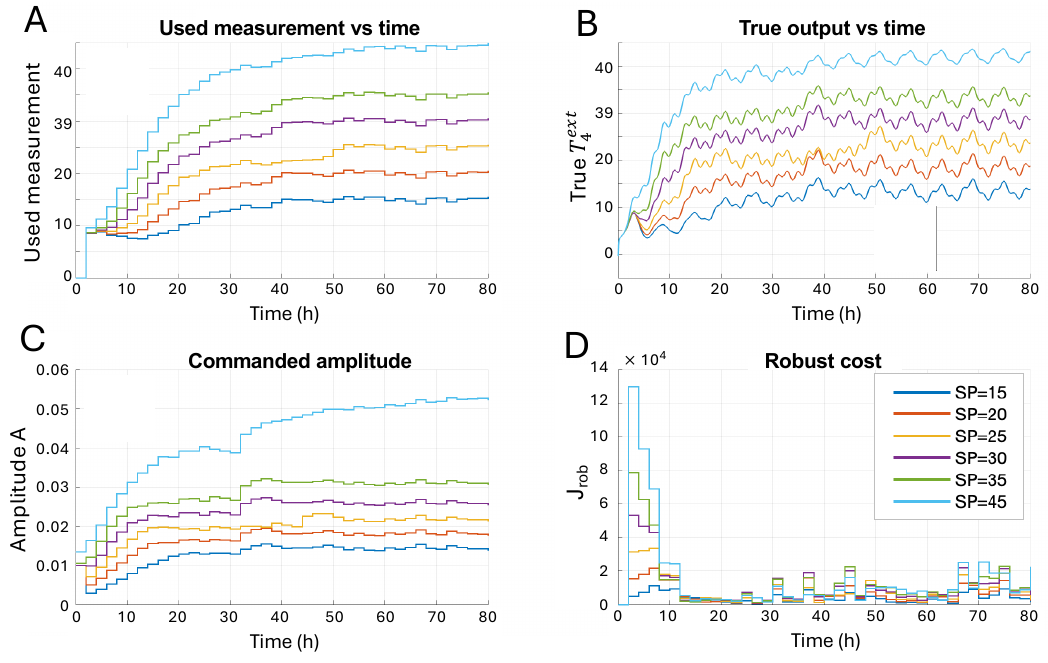}
\caption{Sweep of the RAPID controller over multiple target setpoints \(T_4^{\mathrm{des}}\in\{15, 20,25,30,35,45\}\). Top-left: controller-side used measurement versus time. Top-right: true extracellular output trajectory \(T_4^{\mathrm{ext}}\) versus time. Bottom-left: commanded stimulation amplitude \(A_{\mathrm{cmd}}\). Bottom-right: per-window robust predictive cost \(J_{\mathrm{rob}}\). Across all tested setpoints, the controller increases stimulation in a graded manner, drives the measured and true outputs toward the desired operating range, and then maintains regulation with bounded oscillations under the robust scenario-based adaptation law.}
\label{fig:appendix_rapid_sweep_setpoints}
\end{figure}
\Cref{fig:appendix_rapid_sweep_setpoints} illustrates the behavior of RAPID over a range of prescribed target values. For each setpoint, the controller-side used measurement approaches the desired level, the true output follows with delayed oscillatory transients, and the commanded amplitude adapts to the stimulation level required for sustained regulation. The corresponding robust cost decreases after the initial adaptation phase, indicating improved one-window predictive performance as the closed loop settles. These additional simulations support that the same robust adaptive PID architecture remains effective across multiple target operating points. 

\end{document}